\documentclass[journal,onecolumn]{IEEEtran}

\ifCLASSINFOpdf
\else
\fi

\usepackage{enumitem}

\usepackage{amsthm}
\usepackage{amsfonts,amssymb,latexsym,cite}
\usepackage[cmex10]{amsmath}
\usepackage{algorithmic}
\usepackage{array}
\usepackage{bbm}
\usepackage{mathrsfs}
\usepackage{multirow}
\usepackage{verbatim}
\usepackage{color,xcolor}
\usepackage{graphicx}
\usepackage{epstopdf}
\usepackage{subcaption}
\usepackage[font={small,it}]{caption}
\usepackage{multirow}

\usepackage[T1]{fontenc}
\usepackage{url}
\usepackage{hyperref} 
\graphicspath{{./figures/}} 
\usepackage{caption}
\usepackage{stfloats} 
\usepackage[blocks]{authblk}
\usepackage{epsfig}
\usepackage{latexsym}
\usepackage{soul}
\newtheorem{lemma}{Lemma}
\newtheorem{definition}{Definition}
\newtheorem{theorem}{Theorem}

\newtheorem{claim}{Claim}

\newtheorem{example}{Example}
\newtheorem{remark}{Remark}

\newcommand{\upperRomannumeral}[1]{\uppercase\expandafter{\romannumeral#1}}

\usepackage{lipsum}

\DeclareMathOperator*{\argmax}{arg\,max}

\newcommand{\Rbar}{\bar{\mathsf{R}}_{\mathsf{q},\delta}}

\newcommand{\T}{\mathcal{T}}
\newcommand{\B}{\mathcal{B}}

\newcommand{\PP}{\mathbb{P}}
\newcommand{\DD}{\mathbb{D}}

\newcommand{\M}{\mathcal{M}}
\newcommand{\K}{\mathcal{K}}
\newcommand{\x}{\mathbf{x}}
\newcommand{\y}{\mathbf{y}}
\newcommand{\z}{\mathbf{z}}
\newcommand{\X}{\mathbf{X}}
\newcommand{\Y}{\mathbf{Y}}
\newcommand{\Z}{\mathbf{Z}}
\newcommand{\ux}{\underline{\mathbf{x}}}
\newcommand{\uy}{\underline{\mathbf{y}}}
\newcommand{\uz}{\underline{\mathbf{z}}}

\newcommand{\pyx}{P_{\kern-0.17em Y \kern -0.06em|\kern -0.1em X}}
\newcommand{\pyhx}{P_{\kern-0.17em \hat{Y} \kern -0.06em|\kern -0.1em X}}
\newcommand{\tyx}{T_{\kern-0.17em Y \kern -0.06em|\kern -0.1em X}}
\newcommand{\tyhx}{T_{\kern-0.17em \hat{Y} \kern -0.06em|\kern -0.1em X}}
\newcommand{\wyx}{W_{\kern-0.17em Y \kern -0.06em|\kern -0.1em X}}
\newcommand{\wyxa}{W_{\kern-0.17em Y \kern -0.06em|\kern -0.1em X=0}}
\newcommand{\wyxb}{W_{\kern-0.17em Y \kern -0.06em|\kern -0.1em X=1}}
\newcommand{\wyxn}{W^{\kern -0.06em \otimes n}_{\kern-0.17em Y \kern -0.06em|\kern -0.1em X}}
\newcommand{\wyxw}{W^{\kern -0.06em \otimes w}_{\kern-0.17em Y \kern -0.06em|\kern -0.1em X}}
\newcommand{\wzx}{W_{\kern-0.17em Z \kern -0.06em|\kern -0.1em X}}
\newcommand{\wzxn}{W^{\kern -0.06em \otimes n}_{\kern-0.17em Z \kern -0.06em|\kern -0.1em X}}
\newcommand{\wzxw}{W^{\kern -0.06em \otimes w}_{\kern-0.17em Z \kern -0.06em|\kern -0.1em X}}
\newcommand{\Pj}{P_{Y\widehat{Y}|X}}
\newcommand{\Sq}{\mathcal{S}_\mathsf{q}}

\newcommand{\A}{\mathcal{A}}

\newcommand{\D}{\mathcal{D}}
\newcommand{\Vh}{\widehat{V}}
\newcommand{\yh}{\hat{y}}
\newcommand{\yyh}{\hat{\y}}

\newcommand{\C}{\mathcal{C}}

\newcommand{\pwx}{P^{w}_{\underline{\mathbf{X}}}}
\newcommand{\pwy}{P^{w}_{\underline{\mathbf{Y}}}}
\newcommand{\pwz}{P^{w}_{\underline{\mathbf{Z}}}}

\newcommand{\ppmx}{P_{\X}^{n,l}}

\newcommand{\ppmz}{P_{\Z}^{n,l}}

\allowdisplaybreaks[3]

\begin{document}
%
\title{Covert Communication With Mismatched Decoders}
%
%
%
\author{
        Qiaosheng Zhang,
        Vincent~Y.~F.~Tan,~\IEEEmembership{Senior~Member,~IEEE}
\thanks{Qiaosheng Zhang is with the Department
of Electrical and Computer Engineering, National University of Singapore (e-mail: elezqiao@nus.edu.sg).} 
\thanks{Vincent~Y.~F.~Tan is with the Department
of Electrical and Computer Engineering and Department of Mathematics, National University of Singapore (e-mail: vtan@nus.edu.sg).}
}

%
%

\markboth{}%
{Shell \MakeLowercase{\textit{et al.}}: Bare Demo of IEEEtran.cls for IEEE Journals}
%



\maketitle

\begin{abstract}
This paper considers the problem of covert communication with mismatched decoding, in which a sender wishes to reliably communicate with a receiver whose decoder is fixed and possibly sub-optimal, and simultaneously to ensure that the communication is covert with respect to a warden. We present single-letter lower and upper bounds on the information-theoretically optimal throughput as a function of the given decoding metric,  channel laws, and the desired level of covertness. These bounds match for a variety of scenarios of interest, such as (i) when the channel between the sender and receiver is a binary-input binary-output channel, and (ii) when the decoding metric is particularized to the so-called erasures-only metric. The lower bound is obtained based on a modified random coding union bound with pulse position modulation (PPM) codebooks, coupled with a non-standard expurgation argument. The proof of the upper bound relies on a non-trivial combination of analytical techniques for the problems of covert communication and mismatched decoding.

\end{abstract}

\flushbottom

\section{Introduction}

In contrast to classical information-theoretic security problems that are concerned with hiding the \emph{content} of information, the problem of \emph{covert communication}\footnote{Covert communication is also known as low probability of detection (LPD) communication in the literature.} instead aims to hide the \emph{fact} that communication is taking place. Covert communication has potential applications in a variety of important scenarios, such as military communications. For example, the communication between two submerged submarines should be covert when an enemy maritime patrol aircraft is present, otherwise serious consequences may 
occur.

Due to its potentially widespread applications, covert communication has attracted significant attention in recent years. The pioneering work by Bash \emph{et al.}~\cite{bash2013limits} first demonstrated a \emph{square-root law} for covert communication, stating that one can only covertly and reliably transmit $\Theta(\sqrt{n})$ bits of message over $n$ channel uses. Building upon~\cite{bash2013limits}, subsequent works further characterized  information-theoretic limits of covert communication for diverse channel models, as well as developed low-complexity covert communication schemes by exploiting a variety of coding techniques. We refer the readers to Subsection~\ref{sec:related} below for a detailed literature review.   

While the problem of covert communication has been extensively studied, most, if not all, of the prior works focused on the setting in which the encoder and decoder can be optimized according to the channel law. However, in some practical scenarios, one may not have accurate knowledge about the channel (e.g., submarines often only have an imperfect channel estimation in the ocean). In other scenarios, even if the channel is precisely known, one may still wish to implement a sub-optimal decoder due to computational complexity considerations. Motivated by these practical considerations, we consider the problem of covert communication with \emph{mismatched decoding}~\cite{csiszar1995channel,merhav1994information,csiszar1981graph,hui1983fundamental}, where the decoder is fixed {\em a priori} and possibly sub-optimal. It is assumed that the decoding rule  is governed by a \emph{given} decoding metric, and the only freedom for designers is to optimize the codebook and the encoder. Given the differences between the current problem and standard covert communication problem, it is then natural to ask (i) \emph{What is the highest rate at which message bits can be transmitted covertly and reliably with mismatched decoding (which is formally referred to as the \emph{covert mismatch capacity})}, and (ii) \emph{In which case using a mismatched decoder also lead to an optimal throughput?}  

To address the aforementioned questions, this work investigates the following covert communication setting. The sender occasionally communicates with the legitimate receiver through a binary-input\footnote{It is also possible to consider a more general setting with multiple non-zero input symbols (by following the lead of~\cite{wang2016fundamental}); however, for simplicity and ease of presentation, we focus on the binary-input setting in this work.} discrete memoryless channel (BDMC), and it is assumed that there is a warden who can eavesdrop their communication through another independent BDMC. The goals are twofold. One one hand, the receiver should be able to reliably reconstruct the message by using the given decoding metric. On the other hand, the covertness constraint requires the warden to be unable to determine whether or not communication is taking place. More specifically, we require that at the warden's side, the output distribution when communication takes place is almost indistinguishable from the output distribution when no communication takes place, where the discrepancy between the two output distributions is measured by the \emph{Kullback-Leibler (KL) divergence}. 

The main contributions of this work can be summarized as follows.
\begin{itemize}
    \item We first develop an achievability scheme and derive a lower bound  on the covert mismatch capacity (Theorem~\ref{thm:achievability}). Our scheme is based on pulse position modulation (PPM) codebooks and a careful expurgation argument. 
    
    \item We also provide a single-letter upper bound on the covert mismatch capacity (Theorem~\ref{thm:converse}), which improves on the trivial upper bound---the \emph{covert capacity}~\cite{bloch2016covert, wang2016fundamental}. 
    
    \item When the given decoding metric matches the channel law (which is referred to as the \emph{matched case}), our lower and upper bounds both equal the covert capacity. 
    
    \item When the channel between the sender and receiver is a binary-input  binary-output channel, the lower and upper bounds coincide and thus we have an exact characterization of the covert mismatch capacity (Theorem~\ref{thm:binary}). It is also worth noting that in this case, covert mismatch capacity exhibits a dichotomy---it equals either the covert capacity or zero.
    
    \item Finally, we apply our lower and upper bounds to the problem of \emph{covert communication with zero undetected error}, which is a special case of covert communication with mismatched decoding (by choosing an appropriate decoding metric). The zero undetected error problem is a classical  information theory problem, and requires that the decoder should never output an incorrect message, and that the probability of \emph{erasure} (i.e., detected error) should tend to zero. Perhaps surprisingly, our lower and upper bounds coincide for covert communication over BDMCs  (Theorem~\ref{thm:eo}), and thus we have an exact single-letter expression for the so-called \emph{covert erasures-only capacity}\footnote{The covert erasures-only capacity is defined as the highest rate at which the covertness could be guaranteed, the probability of undetected error exactly equals zero, and the probability of erasure could vanish.}. In contrast, to the best of our knowledge, there does not exist a \emph{computable} capacity expression for standard (non-covert) communication over BDMCs (while an incomputable expression was given in~\cite{csiszar1995channel}).
\end{itemize}
 

\subsection{Technical challenges and solutions}

Unlike the standard mismatched decoding problem, the requirement of covertness puts forth new challenges in designing achievability schemes and proving new coding theorems. We first discuss two challenges from the achievability’s perspective. 
\begin{itemize}
    
    \item First, our lower bound is established using a low-weight PPM codebook (rather than a more common constant composition codebook), due to the requirement of covertness. The PPM codebook can be viewed as a highly structured sub-class of constant composition codebooks, and its optimality for covert communication was first derived by Bloch and Guha~\cite{bloch2017optimal}. 
    While the PPM codebook is ideal from the perspective of covertness, the use of it also raises issues in the reliability part under mismatched decoding. For example, a standard technique to circumvent the non-independent issue of constant composition codebooks  is to approximate the probability of each constant composition codeword by the probability of its corresponding i.i.d. codeword (see~\cite[Chapter 2.6.5]{scarlett2020information} for an example). However, this technique does not apply here because the probability of each PPM codeword is significantly larger than its corresponding i.i.d. codeword. Another issue is that the fixed decoding rule prevents us from designing decoders that have been shown to be suitable for low-weight codebooks (e.g., the modified information-density decoder~\cite{bloch2016covert,tahmasbi2018first,zhang2021covert}). To overcome these issues, we exploit specific properties of the PPM codebook to obtain a \emph{chunk-wise independent structure}, which allows us to decompose the blocklength $n$ into $\Theta(\sqrt{n})$ disjoint chunks and then analyze each chunk independently (see Eqns.~\eqref{eq:xiao}-\eqref{eq:13} for details). 
    
    \item Second, in the context of covert communication, the two common reliability criteria---\emph{average} and \emph{maximum} probability of error---cannot be simply connected through the standard expurgation technique\footnote{The standard expurgation technique states that for any code of size $|\mathcal{M}|$ with a vanishing average probability of error, one can simply expurgate $|\mathcal{M}|/2$ codewords that have highest probability of error to obtain a new code that has a vanishing maximum probability of error.}. This is because although expurgating codewords is helpful from the perspective of reliability, it changes the output distribution induced by the code and thus the resultant code may no longer  satisfy the covertness constraint. In this work, we adopt the more stringent maximum probability of error as the reliability criterion. To ensure a vanishing maximum probability of error, we use a recently developed result by Tahmasbi and Bloch~\cite{tahmasbi2018first} to show the existence of a code such that  \emph{every} subset of codewords (with a fixed cardinality) satisfies the covertness constraint. This allows us to apply an expurgation argument to the set of ``bad'' codewords and to simultaneously ensure the resultant code is still covert (see Remark~\ref{remark:expurgation} for a detailed discussion).
\end{itemize}

The proof of the upper bound requires a  non-trivial combination of analytical techniques for the problems of covert communication and mismatched decoding. We first use an expurgation argument to show that for any code satisfying the covertness constraint, there must exist a low-weight constant composition subset of codewords whose size is almost as large as the original code. Our next step, which analyzes the probability of error of the resultant subset, is inspired~\cite{kangarshahi2021single}, in which the authors established, for the first time, a single-letter upper bound for the mismatched decoding problem. The main idea is to translate the mismatched-decoding error of the original channel to the maximum-likelihood decoding error of an auxiliary channel. However, their key result~\cite[Theorem 5]{kangarshahi2021single} is not readily applicable to our setting due to the stringent input constraint imposed by the covertness requirement. We circumvent this difficulty by proving a strengthened lemma  (Lemma~\ref{lemma:key}) that allows us to lower bound the probability of error for any \emph{low-weight} constant composition code. A detailed comparison between our strengthened lemma and the original result~\cite[Theorem 5]{kangarshahi2021single} is presented in Section~\ref{sec:converse}, right after Lemma~\ref{lemma:key}. 




\subsection{Related works} \label{sec:related}


The theoretical underpinnings of covert communication have been extensively studied following the pioneering work by Bash \emph{et al.}~\cite{bash2013limits}. Researchers have progressively established information-theoretic limits of covert communication for a variety of channel and network models, including discrete memoryless channels~\cite{CheBJ:13, che_reliable_2014-2, che2014reliable, tahmasbi2018first,tahmasbi2017error, bloch2016covert, wang2016fundamental}, Gaussian channels~\cite{bloch2016covert, wang2016fundamental, yan2019gaussian}, multiuser channels~\cite{arumugam2019covert,arumugam2019embedding, tan2018time, kibloff2019embedding,cho2020treating}, channel with states~\cite{lee2018covert, zivarifard2019keyless}, channel with jammers~\cite{sobers2017covert, song2020stealthy, zhang2021covertadv,zivarifard2021covert},  Rayleigh-fading channels~\cite{tahmasbi2019covert2,zheng2020covert,hu2019covert}, continuous-time channels~\cite{wang2019gaussian, wang2021covert2, zhang2019undetectable}, quantum channels~\cite{bash2015quantum, wang2016optimal,gagatsos2020covert}, MIMO channels~\cite{wang2021covert}, adhoc networks~\cite{cho2020throughput,im2020mobility}, etc. In addition to characterizing the information-theoretic limits, researchers have also studied covert communication from a coding perspective, and have investigated various coding techniques such as concatenated codes~\cite{zhang2019covert}, PPM~\cite{bloch2017optimal}, multilevel coding with PPM~\cite{kadampot2020multilevel,kadampot2019codes,wang2021explicit}, and polar codes~\cite{freche2018polar}. In recent years, the concept of covertness has also been incorporated into other research fields, ranging from information theory to information security and wireless communications. Specific topics include secret key agreement~\cite{tahmasbi2019covert}, source coding~\cite{chou2021universal}, identification-via-channels~\cite{zhang2021covert}, authentication~\cite{chang2021covert},  unmanned aerial vehicle~\cite{zhou2021uav,yan2021optimal,xiaobo2021three}, etc.

The mismatched decoding problem~\cite{csiszar1995channel,merhav1994information,csiszar1981graph,hui1983fundamental} is a classical fiendishly hard problem in information theory, and the objective is to understand the highest communication rate when the decoding rule is fixed and possibly
sub-optimal. It is closely related to other long-standing problems such as the zero-error capacity. While multiple lower and upper bounds have been developed, the capacity of the mismatched decoding problem still remains open. The most notable single-letter lower bound is the so-called \emph{LM rate}, which was first derived by Csiszár and Körner~\cite{csiszar1981graph} and Hui~\cite{hui1983fundamental} based on constant composition codes. Csiszár and Narayan~\cite{csiszar1995channel} later showed that the multi-letter version of the LM rate is in general better than its single-letter 
counterpart. Another line of works studied mismatched decoding for multiuser settings~\cite{lapidoth1996mismatched,somekh2014achievable,scarlett2016multiuser}; in particular, Lapidoth  showed that the LM rate can be improved by treating the point-to-point channel as a multiple access channel and using a multiuser scheme~\cite{lapidoth1996mismatched}. While there have been extensive studies on lower bounds, until recently, much less has been understood about upper bounds. Kangarshahi and Guillén i Fàbregas~\cite{kangarshahi2021single} recently provided a single-letter upper bound on the mismatched capacity, and Somekh-Baruch presented both single-letter and multi-letter upper bounds in a series of works~\cite{somekh2018converse, somekh2021single,somekh2021robust}. These recently developed upper bounds significantly promote the understanding of mismatched decoding. We refer the readers to~\cite{scarlett2020information} for a comprehensive survey of the mismatched decoding problem.    

\subsection{Outline}
The rest of this paper is organized as follows. We provide some notational conventions and preliminaries in Section~\ref{sec:notation}, and formally introduce the problem of covert communication with mismatched decoding in Section~\ref{sec:model}. In Section~\ref{sec:result}, we present lower and upper bounds on the covert mismatch capacity, and further provide analytical and numerical evaluations of these bounds for a variety of scenarios of interest. Finally, Sections~\ref{sec:achievability} and~\ref{sec:converse} respectively provide the detailed proofs of the lower and upper bounds.

\section{Preliminaries} \label{sec:notation}
\subsection{Notation}
Random variables and their realizations are respectively denoted by uppercase and lowercase letters, e.g., $X$ and $x$. Sets are denoted by calligraphic letters, e.g., $\mathcal{X}$. Vectors are denoted by boldface letters, e.g., $\X$ or $\x$, where the length of each vector will be clear from the context. We use $X_i$ or $x_i$ to denote the $i$-th element of the vector $\X$ or $\x$, and $X_a^b$ or $x_a^b$ to denote the subsequence $(X_a, X_{a+1}, \ldots, X_b)$ or $(x_a, x_{a+1}, \ldots, x_b)$. Let $\text{wt}_{\mathrm{H}}(\x)$ be the {\em Hamming weight}, or number of non-zero elements, of the vector  $\x$.

In our calculations, logarithms $\log$ and exponentials $\exp$ are to the natural base  $e$. For any real number $c \in \mathbb{R}$, we define $[c]^+ \triangleq \max\{0,c\}$. For any probability distribution $P$ over the finite set $\mathcal{X}$, we denote its $n$-letter product distribution by $P^{\otimes n}$, and the largest probability and smallest non-zero probability respectively by 
\begin{align}
[P]^{\max} \triangleq \max_{x \in \mathcal{X}}P(x) \quad \mathrm{and} \quad [P]^{\min} \triangleq \min_{x \in \mathcal{X}:P(x)>0}P(x).    
\end{align}
For any two probability distributions $P$ and $Q$ over the same finite set $\mathcal{X}$, their KL divergence and $\chi_2$-distance are respectively given by $\DD(P \Vert Q) \triangleq \sum_{x \in \mathcal{X}} P(x) \log \frac{P(x)}{Q(x)}$ and $\chi_2(P \Vert Q)  \triangleq \sum_{x \in \mathcal{X}} \frac{(P(x)-Q(x))^2}{Q(x)}$.
We say $P$ is \emph{absolutely continuous} with respect to $Q$ (denoted by $P \ll Q$) if the support of $P$ is a subset of the support of $Q$ (i.e., for all $x \in \mathcal{X}$ such that $Q(x)=0$, $P(x)=0$).

\subsection{Preliminaries on the method of types}
Given a length-$n$ vector $\x \in \mathcal{X}^n$, we define its \emph{type} (or \emph{empirical distribution}) as $T_{\x}(x) \triangleq \frac{1}{n} \sum_{j=1}^n \mathbbm{1} \{x_i = x\}$. The \emph{type class} corresponding to a specific type $P \in \mathcal{P}_n(\mathcal{X})$ is denoted by $\T_P \triangleq \{\x \in \mathcal{X}^n: T_{\x} = P \}$.  Given two sequences $\x \in \mathcal{X}^n$ and $\y \in \mathcal{Y}^n$, we define their \emph{joint type} as $T_{\x,\y}(x,y) \triangleq \frac{1}{n} \sum_{j=1}^n \mathbbm{1} \{x_j = x, y_j = y\}$, and the \emph{conditional type} of $\y$ given $\x$ as 
\begin{align}
T_{\y|\x}(y|x) \triangleq \begin{cases}
\frac{T_{\x,\y}(x,y)}{T_{\x}(x)}, &\mathrm{if} \ T_{\x}(x) > 0, \\
\frac{1}{|\mathcal{Y}|}, &\mathrm{otherwise}.
\end{cases}
\end{align}
For a given $\x \in \T_P$ and a conditional distribution $V \in \mathcal{P}(\mathcal{Y}|\mathcal{X})$, the set of $\y \in \mathcal{Y}^n$ such that $(\x,\y)$ has joint type $P \times V$ is denoted by $\T_{V}(\x) \triangleq \{\y \in \mathcal{Y}^n: T_{\x,\y} = P \times V\}$. Let $\mathcal{V}_{P}(\mathcal{Y}|\mathcal{X})$ be the set of all $V \in \mathcal{P}(\mathcal{Y}|\mathcal{X})$ for which the conditional type class of a sequence of type $P$ is non-empty.

\section{Problem setting}
\label{sec:model}


\subsection{Model}
The sender may occasionally communicate with the receiver through a binary-input discrete memoryless channel (BDMC) $(\mathcal{X}, \wyx,\mathcal{Y})$, where the input alphabet  $\mathcal{X} = \{0,1\}$ with `$0$' being the \emph{innocent symbol}, and the output alphabet $\mathcal{Y}$ is assumed to be finite. The \emph{transmission status} of the sender is denoted by a binary-valued variable $\Lambda \in \{0,1\}$: 
\begin{itemize}
    \item When $\Lambda = 1$, the sender sends a  message $M$ (which is uniformly chosen from the message set $\mathcal{M}$) to the receiver;
    \item  When $\Lambda=0$, the sender always sends the innocent symbol `$0$' to the channel $\wyx$.
\end{itemize}
The communication between the sender and receiver is possibly assisted by a shared key $K$, which is uniformly distributed over the key set $\mathcal{K}$. There is also a warden who can eavesdrop the communication through another BDMC $(\mathcal{X}, \wzx,\mathcal{Z})$, where $\mathcal{Z}$ is a finite alphabet. The warden does not know the shared key. For notational convenience, we further define
\begin{align}
&P_0 \triangleq W_{Y|X=0}, \quad P_1 \triangleq W_{Y|X=1}, \\
&Q_0 \triangleq W_{Z|X=0}, \quad Q_1 \triangleq W_{Z|X=1}.
\end{align}
Following the convention in the covert communication literature, we make three assumptions on the channels $\wzx$ and $\wyx$: (A1) $Q_0 \ne Q_1$, (A2) $Q_1 \ll Q_0$ (i.e., $Q_1$ is absolutely continuous with respect to $Q_0$), and (A3) $P_1 \ll P_0$. The first two assumptions (A1) and (A2) respectively preclude the scenarios in which covertness is always guaranteed or would never be guaranteed. Without assumption (A3), the sender and receiver are able to communicate $\Theta(\sqrt{n}\log n)$ bits reliably and covertly, breaking the square-root law~\cite[Theorem 7]{bloch2016covert}.

\subsection{Code and Mismatched decoder}
A \emph{code} $\C$ of blocklength $n$  consists of a message set $\mathcal{M}$, a shared key set $\mathcal{K}$, a collection of length-$n$ codewords $\{\x(m,k)\}_{m \in \mathcal{M}, k \in \mathcal{K}}$ (called the \emph{codebook}), and a encoder $f: \mathcal{M} \times \mathcal{K} \to \mathcal{X}^n$ that maps the message-key pair $(m,k)$ to $\x(m,k)$. 
The channel laws corresponding to $n$ channel uses are denoted by $\wyxn (\y|\x) \triangleq \prod_{i=1}^n \wyx(y_i|x_i)$ and $\wzxn (\z|\x)\triangleq \prod_{i=1}^n \wzx(z_i|x_i)$.
Upon receiving $\y \in \mathcal{Y}^n$ and based on the knowledge of the shared key $k \in \mathcal{K}$, the decoder outputs the message $\widehat{M}$ such that
\begin{align}
    \widehat{M} = \argmax_{m \in \mathcal{M}} \mathsf{q}^n(\x(m,k),\y), \quad \mbox{where } \mathsf{q}^n(\x(m,k),\y) = \prod_{i=1}^n \mathsf{q}(x_i(m,k), y_i),
\end{align}
and $\mathsf{q}: \mathcal{X} \times \mathcal{Y} \to (0,\infty)$ is the given decoding metric. When there is a tie, the decoder simply declares an error.

\begin{remark}{\em
For simplicity, we consider the setting in which $\mathsf{q}(x,y)$ takes on values on the {\em positive} real line, while some other works may allow $\mathsf{q}(x,y)$ to only be non-negative. This helps us to avoid some complicated special cases that arise in covert communication. For example, if there exists a $y \in \mathcal{Y}$ such that $P_1(y)>0$, $\mathsf{q}(1,y) > 0$ and $\mathsf{q}(0,y) = 0$, the sender and receiver can adopt the scheme described in~\cite[Appendix G]{bloch2016covert} to communicate $\Theta(\sqrt{n}\log n)$ bits of message over $n$ channel uses---this breaks the square-root law for covert communication, and the coding rate (measured according to Definition~\ref{def:rate}) would become infinity.  }
\end{remark}

\subsection{Reliability and covertness  criteria} 
In this work, we adopt the more stringent \emph{maximum probability of error} to measure the reliability of communication.
\begin{definition}[Probabilities of error]
{\em
When transmitting $\x(m,k)$, a decoding error occurs if there exists another codeword $\x(m',k)$ having a higher value $\mathsf{q}^n(\x(m',k),\y)$. The corresponding probability of error $P_{\mathrm{err}}(m,k)$ is
\begin{align}
P_{\mathrm{err}}(m,k) &= \PP_{\wyxn}\Big( \exists m'\ne m: \mathsf{q}^n(\x(m',k),\Y) \ge \mathsf{q}^n(\x(m,k),\Y) \Big).   
\end{align}
The \emph{maximal probability of error} of the code $\C$, which is maximized over all the message-key pairs $(m,k) \in \mathcal{M} \times \mathcal{K}$, is defined as $P_{\mathrm{err}}^{\max}(\C) \triangleq \max_{m \in \mathcal{M},k \in \mathcal{K}} P_{\mathrm{err}}(m,k)$.
}
\end{definition}

As is common in the covert communication literature, we measure the covertness  with respect to the warden via a KL divergence metric. To be specific, when the transmission status $\Lambda = 1$ (i.e., communication is taking place), 
the output distribution at the warden's side is denoted by 
\begin{align}
\widehat{Q}_{\C}^n(\z) \triangleq \frac{1}{|\mathcal{M}||\mathcal{K}|} \sum_{m \in \mathcal{M}} \sum_{k \in \mathcal{K}} \wzxn(\z| \x(m,k)), \quad \forall \z \in \mathcal{Z}^n.
\end{align}
Note that $\widehat{Q}_{\C}^n$ is the distribution induced by the code $\C$ and channel $\wzx$. When the transmission status $\Lambda =0$, the output distribution at the warden's side is $Q_0^{\otimes n}$, since the channel input is always the innocent symbol `$0$'.  
As shown in Definition~\ref{def:covert} below, we require the KL divergence between $\widehat{Q}_{\C}^n$  and $Q_0^{\otimes n}$  to be bounded from above by a \emph{covertness parameter} $\delta > 0$.

\begin{definition}[Covertness]\label{def:covert}
{\em
The code is said to be $\delta$-covert if its output distribution $\widehat{Q}_{\C}^n$ satisfies $\DD\left(\widehat{Q}_{\C}^n \Vert Q_0^{\otimes n} \right) \le \delta.$
}
\end{definition}

\begin{remark}
{\em In addition to the KL divergence metric, another widely used covertness metric is the \emph{variational distance} metric~\cite{bash2013limits,CheBJ:13, che_reliable_2014-2, che2014reliable, tahmasbi2018first}. Both of the two metrics impose stringent constraints on the warden's estimator, and require that the best estimator should not outperform random guessing by too much. Besides, other metrics, such as \emph{the probability of missed detection for fixed probability of false alarm}~\cite{tahmasbi2018first}, have also been studied in the literature. Although this work focuses on the KL divergence metric, our results can also be extended to other metrics.
}
\end{remark}

\subsection{Covert mismatch capacity} Prior works on covert communication have already shown that, due to the stringent covertness constraint, one can only transmit $\Theta(\sqrt{n})$ bits reliably and covertly over $n$ channel uses (i.e., the square-root law). Thus, it is natural to define the rate to be the logarithm of the message size  $|\mathcal{M}|$ normalized by $\sqrt{n}$ (rather than $n$ in non-covert communication). Below, we present the definitions of the \emph{achievable rate pair} and \emph{covert mismatch capacity}.

\begin{definition} \label{def:rate}
{\em
For the given decoding metric $\mathsf{q}$ and covertness parameter $\delta > 0$, a rate pair $(R,R_K)$ is said to be $(\mathsf{q},\delta)$-\emph{achievable} if there exists a sequence of codes with increasing blocklength $n$ such that 
\begin{align}
    &\liminf_{n \to \infty} \frac{\log|\mathcal{M}| }{\sqrt{n}} \ge R, \qquad \ \  \limsup_{n \to \infty} \frac{\log|\mathcal{K}| }{\sqrt{n}} \le R_K, \\
    &\lim_{n \to \infty} \DD\left(\widehat{Q}_{\C}^n \Vert Q_0^{\otimes n} \right) \le \delta,  \ \ \  \lim_{n \to \infty} P_{\mathrm{err}}^{\max}(\C) = 0.
\end{align}
The \emph{covert mismatch capacity} $\mathsf{C}_{\mathsf{q},\delta}$  is defined as  the supremum of $R$ over all $(\mathsf{q},\delta)$-achievable rate pairs.
}
\end{definition}

\vspace{10pt}
\section{Main results} \label{sec:result}
This section comprises the main results of this work. In Subsections~\ref{sec:lower} and~\ref{sec:upper}, we respectively present a lower bound and  an upper bound on the covert mismatch capacity $\mathsf{C}_{\mathsf{q},\delta}$. While these bounds do not match in general, Subsection~\ref{sec:binary} shows that they do match when the channel $\wyx$ is a binary-input binary-output channel, and thus we have an exact characterization of the covert mismatch capacity. In Subsection~\ref{sec:ternary}, we provide numerical evaluations of our bounds when the channel $\wyx$ is a binary-input ternary-output channel. In Subsection~\ref{sec:zero}, we introduce the problem of covert communication with zero undetected error, in which the capacity can be characterized based on our  bounds derived for mismatched decoding. 

Before introducing the results, we first define the \emph{weight parameter} $t_\delta$ as
\begin{align}
t_\delta \triangleq \sqrt{\frac{2\delta}{\chi_2 (Q_1 \Vert Q_0)}},
\end{align}
whose operational meaning is that the average Hamming weight of codewords in any $\delta$-covert code should not exceed $t_{\delta}\sqrt{n}(1+o(1))$. Note that the weight parameter $t_{\delta}$ is an increasing function of the covertness parameter $\delta$, and a decreasing function of $\chi_2 (Q_1 \Vert Q_0)$. This makes intuitive sense because the average Hamming weight is allowed to be larger if the covertness constraint is less stringent, or the warden's observed output distributions  $Q_0$ and $Q_1$ become closer to each other.

\subsection{Lower bound} \label{sec:lower}
Theorem~\ref{thm:achievability}  below presents a  single-letter lower bound on the covert mismatch capacity $\mathsf{C}_{\mathsf{q},\delta}$.

\begin{theorem}[Lower bound] \label{thm:achievability}
For the given decoding metric $\mathsf{q}$, any covertness parameter $\delta > 0$, and any pair of BDMCs $(\wyx,\wzx)$, we define \begin{align}
    \underline{\mathsf{R}}_{\mathsf{q},\delta} \triangleq t_\delta \cdot \left(\sup_{s \ge 0} \   \mathbb{E}_{P_1}\!\!\left[\log\frac{\mathsf{q}(1,Y)^s}{\mathsf{q}(0,Y)^s}\right] - \log \mathbb{E}_{P_0}\!\!\left[\frac{\mathsf{q}(1,Y)^s}{\mathsf{q}(0,Y)^s} \right] \right). \label{eq:lower}
\end{align}
The rate pair $(\underline{\mathsf{R}}_{\mathsf{q},\delta},R_K)$ is achievable for any $R_K \ge \max\left\{0,\ t_\delta \cdot \DD\left(Q_1 \Vert Q_0 \right) - \underline{\mathsf{R}}_{\mathsf{q},\delta}\right\}.$
Thus, the covert mismatch capacity $\mathsf{C}_{\mathsf{q},\delta} \ge \underline{\mathsf{R}}_{\mathsf{q},\delta}$.
\end{theorem}

The proof of Theorem~\ref{thm:achievability} is presented in Section~\ref{sec:achievability}. Some remarks on Theorem~\ref{thm:achievability} are in order.
\begin{enumerate}
\item When particularizing the decoding metric $\mathsf{q}$ to the channel law $\wyx$ (i.e., $\mathsf{q}(0,0) = P_0(0)$, $\mathsf{q}(0,1) = P_0(1)$, $\mathsf{q}(1,0) = P_1(0)$, and $\mathsf{q}(1,1) = P_1(1)$), it can be shown that $s = 1$ maximizes the right-hand side (RHS) of~\eqref{eq:lower} and thus 
\begin{align}
\underline{\mathsf{R}}_{\mathsf{q},\delta} = t_{\delta} \cdot \mathbb{E}_{P_1}\big[\log (P_1(Y)/P_0(Y))\big] = t_{\delta}\cdot \DD(P_1 \Vert P_0).
\end{align}
Recall that the \emph{covert capacity} for covert communication~\cite{bloch2016covert, wang2016fundamental} is also $\mathsf{C}_{\delta} = t_{\delta} \cdot \DD(P_1 \Vert P_0)$. This means that our lower bound is optimal under the matched case (i.e., the case when the decoding metric matches the channel law). 

\item While the form of $\underline{\mathsf{R}}_{\mathsf{q},\delta}$ in~\eqref{eq:lower} looks like the \emph{GMI rate}~\cite{kaplan1993information,ganti2000mismatched} (compared to the LM rate~\cite{csiszar1981graph, hui1983fundamental, csiszar1995channel, merhav1994information}) in the mismatched decoding literature\footnote{We refer the  readers to~\cite[Section 2.3]{scarlett2020information} for a comprehensive survey of the GMI and LM rates in the mismatched decoding problem. Roughly speaking, the GMI rate is derived using an i.i.d. codebook, while the LM rate is derived using a constant composition codebook. The LM rate is in general larger than the GMI rate.}, $\underline{\mathsf{R}}_{\mathsf{q},\delta}$ is actually more similar to the LM rate since its derivation relies on a special sub-class of constant composition codebooks---the PPM codebook. As we shall see in Subsection~\ref{sec:binary}, when the output alphabet $\mathcal{Y}$ is binary, $\underline{\mathsf{R}}_{\mathsf{q},\delta}$ possesses the same property as the LM rate, i.e., it equals either the covert capacity or zero.
    
    \item The choice of $R_K$ in Theorem~\ref{thm:achievability} ensures that the size of the whole codebook is at least $\exp\{t_{\delta}\DD(Q_1 \Vert Q_0 )\sqrt{n} \}$, which, from a channel resolvability perspective, is the necessary condition for driving $\DD(\widehat{Q}_{\C}^n \Vert Q_0^{\otimes n})$ to be less than $\delta$.   
\end{enumerate}

\subsection{Upper bound} \label{sec:upper}
In the following, we first introduce the definition of \emph{maximal joint conditional distribution} (adapted from~\cite[Definition 1]{kangarshahi2021single}), and then present a  single-letter upper bound on the covert mismatch capacity $\mathsf{C}_{\mathsf{q},\delta}$.  

\begin{definition}[Maximal joint conditional distribution] \label{def:maximal}
{\em
For each pair $(y, \yh)\in \mathcal{Y}\times \mathcal{Y}$, we define
\begin{align}
        \Sq(y,\yh) \triangleq \left\{x \in \mathcal{X} \Big| x = \argmax_{x' \in \mathcal{X}} \frac{\mathsf{q}(x',\yh)}{\mathsf{q}(x',y)} \right\}.
    \end{align}
A joint conditional distribution $\Pj \in \mathcal{P}(\mathcal{Y}\times \mathcal{Y}|\mathcal{X})$ is called a \emph{maximal joint conditional distribution} if $\Pj(y,\yh|x) = 0$ for all $(y, \yh) \in \mathcal{Y} \times \mathcal{Y}$ and $x \notin \Sq(y,\yh)$.  
The set of maximal joint conditional distribution  is defined as
\begin{align}
\mathcal{M}_{\max}(\mathsf{q}) \triangleq \left\{
\Pj \in \mathcal{P}(\mathcal{Y}\times \mathcal{Y}|\mathcal{X}): \Pj \mbox{ is maximal}\right\}.    
\end{align}
}
\end{definition}

\begin{theorem}[Upper bound] \label{thm:converse}
For the given decoding metric $\mathsf{q}$, any covertness parameter $\delta > 0$, and any pair of BDMCs $(\wyx,\wzx)$, we define \begin{align}
\bar{\mathsf{R}}_{\mathsf{q},\delta}  \triangleq t_{\delta} \cdot \left[ \min_{\Pj \in \mathcal{M}_{\max}(\mathsf{q}): P_{Y|X}=\wyx}  \DD\left(P_{\widehat{Y}|X=1} \Vert P_{\widehat{Y}|X=0}\right) \right],   \label{eq:minimizer} 
\end{align}
and  the covert mismatch capacity $\mathsf{C}_{\mathsf{q},\delta} \le \bar{\mathsf{R}}_{\mathsf{q},\delta}$.
\end{theorem}

The proof of Theorem~\ref{thm:converse} is presented in Section~\ref{sec:converse}. From Theorems~\ref{thm:achievability} and~\ref{thm:converse}, it is clear that the covert mismatch capacity satisfies $\underline{\mathsf{R}}_{\mathsf{q},\delta} \le \mathsf{C}_{\mathsf{q},\delta} \le \bar{\mathsf{R}}_{\mathsf{q},\delta}$.
Also note that the covert capacity is a trivial upper bound on the covert mismatch capacity, i.e., $\mathsf{C}_{\mathsf{q},\delta} \le \mathsf{C}_{\delta}$, since in the standard covert communication problem, the decoder can be chosen arbitrarily and thus the capacity can only be higher.
Indeed, the upper bound $\bar{\mathsf{R}}_{\mathsf{q},\delta}$ is in general better than the trivial upper bound $\mathsf{C}_{\delta}$. This is because the set $\mathcal{M}_{\max}(\mathsf{q})$ contains a joint conditional distribution $\widetilde{P}_{Y\widehat{Y}|X}$ whose marginal distributions $\widetilde{P}_{Y|X} = \widetilde{P}_{\widehat{Y}|X} = \wyx$, implying 
\begin{align}
\bar{\mathsf{R}}_{\mathsf{q},\delta}  \le    t_{\delta} \cdot \DD(\widetilde{P}_{\widehat{Y}|X=1} \Vert \widetilde{P}_{\widehat{Y}|X=0}) = t_{\delta} \cdot \DD(P_1 \Vert P_0) = \mathsf{C}_{\delta}.
\end{align}
As a result, the covert mismatch capacity satisfies $\underline{\mathsf{R}}_{\mathsf{q},\delta} \le \mathsf{C}_{\mathsf{q},\delta} \le \bar{\mathsf{R}}_{\mathsf{q},\delta} \le \mathsf{C}_{\delta}$.


\subsection{Covert mismatch capacity for binary-input binary-output channels} \label{sec:binary}

When the channel $\wyx$ between the sender and receiver is a binary-input binary-output channel (i.e., $\mathcal{X} = \mathcal{Y} = \{0,1\}$), Theorem~\ref{thm:binary} below states that the lower bound $\underline{\mathsf{R}}_{\mathsf{q},\delta}$ and upper bound $\bar{\mathsf{R}}_{\mathsf{q},\delta}$ coincide, and thus we have an exact characterization of the covert mismatch capacity. Theorem~\ref{thm:binary} also shows a dichotomy of $\mathsf{C}_{\mathsf{q},\delta}$---it equals either the covert capacity $\mathsf{C}_{\delta}$ or zero. Note that the output alphabet $\mathcal{Z}$ at the warden's side is not necessarily binary.

\begin{theorem} \label{thm:binary}
Let $\mathcal{X} = \mathcal{Y} = \{0,1\}$, and without loss of generality we assume the channel $\wyx$ satisfies $P_0(0) + P_1(1) \ge P_0(1) + P_1(0)$.
\begin{enumerate}
    \item When $\mathsf{q}(0,0)   \mathsf{q}(1,1) > \mathsf{q}(0,1)   \mathsf{q}(1,0)$, the covert mismatch capacity is $\mathsf{C}_{\mathsf{q},\delta} = \underline{\mathsf{R}}_{\mathsf{q},\delta} =  \bar{\mathsf{R}}_{\mathsf{q},\delta} = t_{\delta} \cdot \DD(P_1 \Vert P_0)$,
which also equals the covert capacity $\mathsf{C}_{\delta}$.

\item When $\mathsf{q}(0,0)   \mathsf{q}(1,1) \le \mathsf{q}(0,1)  \mathsf{q}(1,0)$, we have $\mathsf{C}_{\mathsf{q},\delta} = \underline{\mathsf{R}}_{\mathsf{q},\delta} =  \bar{\mathsf{R}}_{\mathsf{q},\delta} = 0$.
\end{enumerate}
\end{theorem}

Theorem~\ref{thm:binary} is proved by evaluating $\underline{\mathsf{R}}_{\mathsf{q},\delta}$ and  $\bar{\mathsf{R}}_{\mathsf{q},\delta}$ for the binary-input binary-output setting, and the details can be found in Appendix~\ref{appendix:binary}.
In the classical mismatched decoding problem (without covertness constraints), it is also known that the mismatch capacity exhibits a dichotomy (i.e., equals either the channel capacity or zero) for binary-input binary-output channels~\cite{csiszar1995channel}. Thus, Theorem~\ref{thm:binary} can be viewed as a counterpart of the aforementioned classical result under the covert communication framework.

\subsection{Numerical evaluations of the lower and upper bounds when $|\mathcal{Y}| = 3$} \label{sec:ternary} 
When the cardinality of the output alphabet $|\mathcal{Y}| \ge 3$, the lower and upper bound do not match in general. In the following, we provide numerical evaluations of $\underline{\mathsf{R}}_{\mathsf{q},\delta}$ and  $\bar{\mathsf{R}}_{\mathsf{q},\delta}$ for a binary-input ternary-output setting.

\begin{example}[Binary-Input Ternary-Output]
{\em
We set the covertness parameter $\delta = 0.1$, and the channels and decoding metric $\mathsf{q}$ to be
\begin{align}
    \wyx = \begin{bmatrix}
0.6 & 0.2 & 0.2\\
0.2 & 0.2 & 0.6
\end{bmatrix}, \qquad  
\mathsf{q} = \begin{bmatrix}
u & 1 & 1\\
1 & 1 & 3
\end{bmatrix}, \quad \mathrm{and} \quad 
\wzx = \begin{bmatrix}
0.8 & 0.1 & 0.1\\
0.2 & 0.3 & 0.5
\end{bmatrix},
\end{align}
where $u \in \mathbb{R}$ is a variable. In Fig.~\ref{fig:2a}, we plot the lower bound $\underline{\mathsf{R}}_{\mathsf{q},\delta}$, upper bound $\bar{\mathsf{R}}_{\mathsf{q},\delta}$, and covert capacity $\mathsf{C}_{\delta}$ as $u$ increases from $0$ to $10$. As expected, the lower bound achieves maximum (and also achieves the covert capacity) when $u = 3$, which corresponds to the matched case where the decoding metric $\mathsf{q}$ is proportional to the channel law $\wyx$. Besides, we note that when $u < 1$, the upper bound $\bar{\mathsf{R}}_{\mathsf{q},\delta}$ is strictly better than the trivial upper bound $\mathsf{C}_{\delta}$.
}
\end{example}

\begin{figure}[t]
		\centering
		\includegraphics[width=7.6cm]{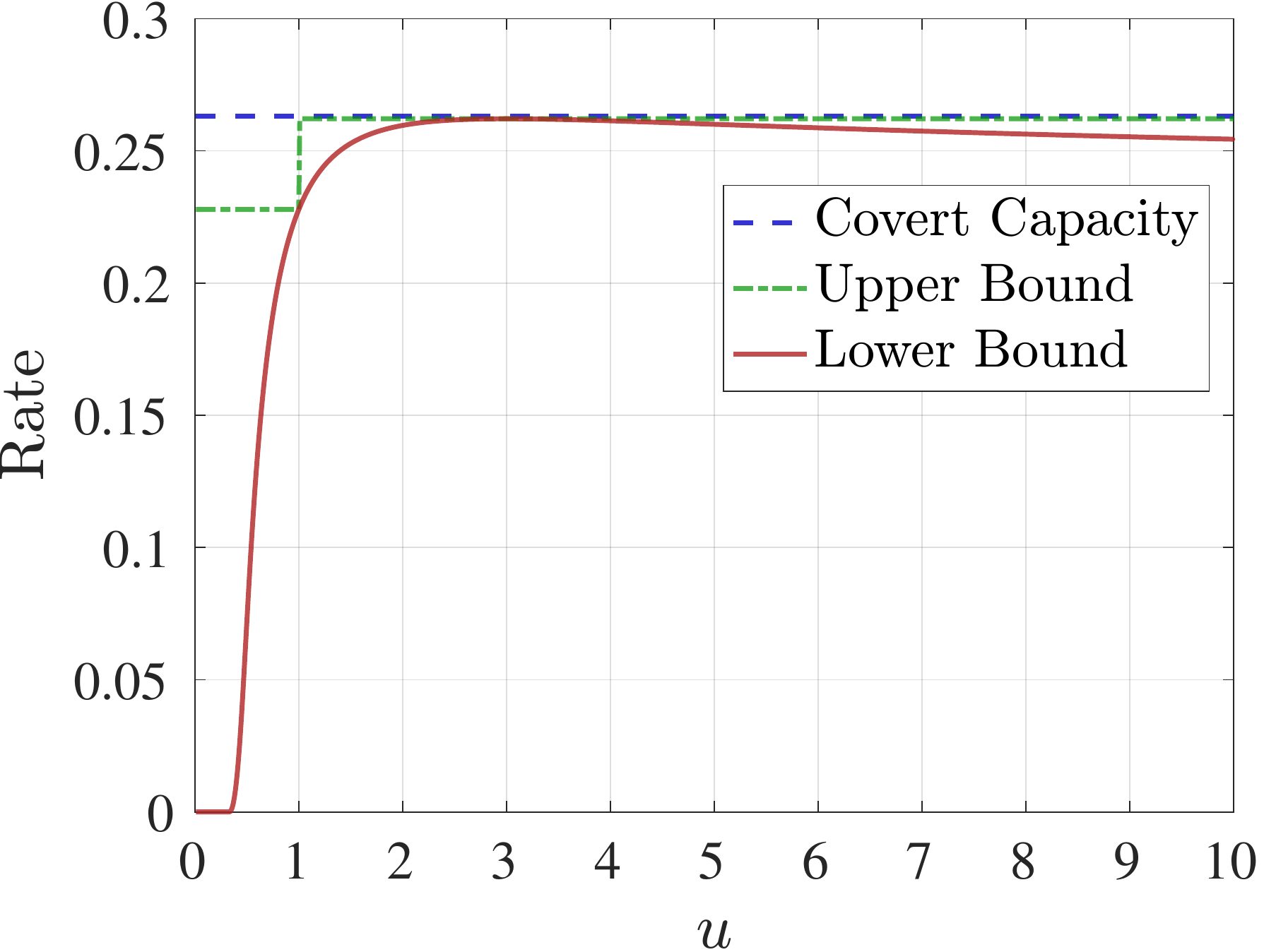}
		\caption{Plots of the lower bound, upper bound, and covert capacity for the binary-input ternary-output setting in Example 1}
		\label{fig:2a}
\end{figure}

\subsection{Covert communication with zero undetected error} \label{sec:zero}

In this subsection, we consider the problem of covert communication with zero undetected error. We say an \emph{undetected error} occurs if the decoder outputs an incorrect message, and an \emph{erasure} occurs if the decoder chooses to declare an error. Since undetected errors are often  more harmful, it is of interest to investigate the highest rate at which the undetected error can be \emph{exactly} zero and the probability of erasure vanishes as the blocklength grows. This problem is referred to as the zero undetected error problem or the erasures-only problem.
It has been noticed~\cite{csiszar1995channel} that the erasures-only problem is a special case of the mismatched decoding problem with decoding metric
\begin{align}
    \mathsf{q}(x,y) = \begin{cases}
    1, &\mathrm{if } \ \wyx(y|x) > 0, \\
    \xi, &\mathrm{if } \ \wyx(y|x) = 0,
    \end{cases} \label{eq:q34}
\end{align}
where $\xi \in (0,1)$ can be chosen arbitrarily.
With this decoding metric, the decoder chooses a codeword $\x$ if it is the only codeword satisfying $\mathsf{q}^n(\x,\y) = 1$ (or equivalently, $\wyxn(\y|\x) > 0$), and declares an error if there are multiple codewords satisfying $\mathsf{q}^n(\x,\y) = 1$. Since the correct codeword $\x$ always satisfies $\wyxn(\y|\x) > 0$, an undetected error would never occur, and the probability of error in this specific mismatched decoding problem is exactly the probability of erasures. Thus, we define the \emph{covert erasures-only capacity} $\mathsf{C}^{(\mathrm{eo})}_{\delta}$ as the covert mismatched capacity when $\mathsf{q}$ is particularized to the erasures-only metric in~\eqref{eq:q34}. 

For any BDMC $\wyx$, we define the set $\mathcal{D}$ as a collection of symbols $y \in \mathcal{Y}$ that can only be induced by the innocent symbol $X = 0$, i.e., $\mathcal{D} \triangleq \left\{y \in \mathcal{Y}: P_0(y) > 0, P_1(y) = 0 \right\}$. Also note that the set $\left\{y \in \mathcal{Y}: P_1(y) > 0, P_0(y) = 0 \right\}$ is empty, since  the channel $\wyx$ satisfies the absolute continuity assumption $P_1 \ll P_0$.

\begin{theorem} \label{thm:eo}
For any covertness parameter $\delta > 0$  and any pair of BDMCs $(\wyx,\wzx)$, the covert erasures-only capacity 
\begin{align}
\mathsf{C}^{(\mathrm{eo})}_{\delta} =   t_{\delta} \cdot \log\left(\frac{1}{P_0(\mathcal{Y}\setminus \mathcal{D})}\right). 
\end{align}
\end{theorem}

Theorem~\ref{thm:eo} is proved by evaluating the lower and upper bounds for the covert mismatched capacity when particularizing the decoding metric to~\eqref{eq:q34}, and the details are deferred to Appendix~\ref{appendix:eo}.
To the best of our knowledge, a computable expression of \emph{erasures-only capacity} for standard (non-covert) communication over a BDMC is not known. In contrast, we show that for covert communication over BDMCs (with the aforementioned absolute continuity assumption), the covert erasures-only capacity  can be precisely characterized by a single-letter expression.    


\section{Proof of Theorem~\ref{thm:achievability} (Lower Bound)} \label{sec:achievability}
We first introduce our code design and provide a short proof sketch, and defer the detailed analyses to Subsections~\ref{sec:ppm}--\ref{sec:expurgate}.
Let the sizes of the message set and key set to be 
\begin{align}
    &\frac{\log |\mathcal{M}|}{\sqrt{n}} = \underline{\mathsf{R}}_{\mathsf{q},\delta} - \eta_1, \quad \mathrm{and}\\
    &\frac{\log |\mathcal{M}| + \log |\mathcal{K}|}{\sqrt{n}} = \max\left\{\underline{\mathsf{R}}_{\mathsf{q},\delta} - \eta_1, t_{\delta} \DD\left(Q_1 \Vert Q_0 \right)+ \eta_2 \right\},
\end{align}
where $\eta_1,\eta_2 \in (0,1)$ can be made  arbitrarily small. Note that no shared key is needed if $\underline{\mathsf{R}}_{\mathsf{q},\delta} - \eta_1 \ge t_{\delta} \DD\left(Q_1 \Vert Q_0 \right)+ \eta_2$ (corresponding to the scenario when the channel $\wzx$ is sufficiently noisy).

For each message-key pair $(m,k) \in \mathcal{M} \times \mathcal{K}$, we generate a codeword $\x(m,k) \in \mathcal{X}^n$ independently according to the so-called \emph{PPM distribution} $\ppmx$ (to be described in Subsection~\ref{sec:ppm} below). The code $\C$ contains all the codewords $\{\x_{m,k} \}$ for $m \in \mathcal{M}$ and $k \in \mathcal{K}$. For every $k' \in \mathcal{K}$,  the corresponding sub-code $\C_{k'}$ contains all the codewords $\{\x(m,k') \}_{m \in \mathcal{M}}$ indexed by $k'$, and its average probability of error $P^{\mathrm{(avg)}}_{\mathrm{err}}(\C_k)$ and maximal probability of error $P^{\mathrm{(max)}}_{\mathrm{err}}(\C_k)$ are respectively given by
\begin{align}
&P^{\mathrm{(avg)}}_{\mathrm{err}}(\C_k) = \frac{1}{|\mathcal{M}|} \sum_{m \in \mathcal{M}} P_{\mathrm{err}}(m,k),  \\
&P^{\mathrm{(max)}}_{\mathrm{err}}(\C_k) = \max_{m \in \mathcal{M}} P_{\mathrm{err}}(m,k).
\end{align}
Also note that $P^{\mathrm{(max)}}_{\mathrm{err}}(\C) = \max_{k \in \mathcal{K}}P^{\mathrm{(max)}}_{\mathrm{err}}(\C_k)$.

\underline{Proof Sketch}:
While the ultimate goal is to show that $\lim_{n \to \infty} P^{\mathrm{(max)}}_{\mathrm{err}}(\C) = 0$ (i.e., the maximal probabilities of error of \emph{all} sub-codes are vanishing), we instead choose to analyze the average probability of error of each sub-code as our first step, by using a modified random coding union (RCU) bound. The details are provided in Subsection~\ref{sec:average}.

Based on a random coding argument, in Subsection~\ref{sec:existence} we show the existence of a code $\C = \{\C_k\}_{k \in \mathcal{K}}$ such that 
\begin{enumerate}
    \item The average probabilities of error $P^{\mathrm{(avg)}}_{\mathrm{err}}(\C_k)$ of \emph{all} sub-codes are vanishing.  
    \item Every subset of the codebook\footnote{With a slight abuse of notation, we use $\C$ to denote both the code and the codebook.} $\mathcal{I} \subset \C$ with cardinality $|\mathcal{I}| = n^{-11}|\C|$ satisfies the ``resolvability'' property---its induced output distribution (i.e., the distribution induced by the uniformly distributed codewords in $\mathcal{I}$ and the channel $\wzx$) is close to the distribution $\ppmz$ (which, to be defined in~\eqref{eq:ppmz}, is induced by the PPM distribution $\ppmx$ and  channel $\wzx$).
\end{enumerate}
The second property is critical for our expurgation argument described below, since it ensures that \emph{every} subset of the code $\C$ satisfies the desired ``resolvability'' property (which is the key for achieving covertness).

We then construct a new code $\widetilde{\C} = \{\widetilde{\C}_k\}_{k \in \mathcal{K}}$ from the original code $\C= \{\C_k\}_{k \in \mathcal{K}}$, by expurgating $n^{-11}|\mathcal{M}|$ of the codewords that have the highest probabilities of error from each of the sub-code $\{\C_k\}_{k \in \mathcal{K}}$. We show in Subsection~\ref{sec:expurgate} that the new expurgated code $\widetilde{\C}$ has the following desired properties: 
\begin{enumerate}
    \item The maximal probabilities of error $P^{\mathrm{(max)}}_{\mathrm{err}}(\widetilde{\C}_k)$ of \emph{all} sub-code $\widetilde{\C}_k$ are vanishing;
    \item The output distribution $\widehat{Q}^n_{\widetilde{\C}}$ induced by $\widetilde{\C}$ satisfies the covertness constraint, i.e., $\DD(\widehat{Q}^n_{\widetilde{\C}} \Vert Q_0^{\otimes n}) \le \delta$;
    \item The message size of each sub-code $\widetilde{\C}_k$ is almost as large as that of $\C_k$, i.e.,
    \begin{align}
        \liminf_{n \to \infty} \frac{\log |\widetilde{\C}_k|}{\sqrt{n}} = \liminf_{n \to \infty} \frac{\log|\C_k|}{\sqrt{n}} = \liminf_{n \to \infty} \frac{\log|\mathcal{M}|}{\sqrt{n}} = \underline{\mathsf{R}}_{\mathsf{q},\delta} -\eta_1, 
    \end{align}
    where $\eta_1 > 0$ can be made arbitrarily small.
\end{enumerate}

\begin{remark} \label{remark:expurgation} {\em
It is worth noting that applying the expurgation technique makes the proof of covertness challenging. This is because although it is relatively standard to prove that the output distribution of the original code $\C$ satisfies the covertness constraint (e.g., through a channel resolvability argument), it is more challenging to do so for the expurgated code $\widetilde{\C}$ since its output distribution differs from that of the original code $\C$ (due to the expurgation process). To solve this issue, we use a recently developed result  in~\cite{tahmasbi2018first} showing that for \emph{every} subset of the original code $\C$ with a fixed cardinality, its output distribution  satisfies the covertness constraint (as shown in Lemma~\ref{lemma:subset} below). This eventually implies that the expurgated code $\widetilde{\C}$  satisfies the covertness constraint.  }
\end{remark}

\subsection{Pulse position modulation (PPM)} \label{sec:ppm}
We now formally introduced the PPM  distribution $\ppmx$ used in our code design. Let the Hamming weight of each codeword in the codebook be $l \triangleq\left \lfloor \sqrt{ \frac{(2\delta-n^{-1/3})n}{\chi_2(Q_1\Vert Q_0)} }  \right\rfloor$, and note that $\lim_{n \to \infty} l/t_{\delta}\sqrt{n} = 1$. We also define 
 $(w,s)$ as non-negative integers such that $w \triangleq \left \lfloor{n/l}\right \rfloor$ and $r \triangleq n - wl$. We use $\ux \in \mathcal{X}^w, \uy \in \mathcal{Y}^w, \uz \in \mathcal{Z}^w$ to denote vectors of length $w$. Let 
\begin{align}
\pwx(\ux) \triangleq \begin{cases}
1/w, & \text{if} \ \text{wt}_{\mathrm{H}}(\ux) = 1, \\
0, &\text{otherwise},
\end{cases} \notag
\end{align}
be the distribution on $\mathcal{X}^w$ such that $\pwx(\ux)$ is non-zero if and only if $\ux$ has Hamming weight one. 
For each length-$n$ vector $\x$ and each $i \in [1:l]$, we define $\ux^{(i)} \triangleq x_{(i-1)w+1}^{iw}$ as the length-$w$ subsequence that  comprises consecutive elements from $x_{(i-1)w+1}$ to $x_{iw}$. Thus, $\x$ can be represented as $\x = [\ux^{(1)}, \ldots, \ux^{(l)}, x_{wl+1}^n]$, where $x_{wl+1}^n$ is of length $r$. The PPM distribution is defined as 
\begin{align}
\ppmx(\x) \triangleq \prod_{i=1}^l \pwx(\ux^{(i)}) \cdot \mathbbm{1}\left\{\text{wt}_{\mathrm{H}}(x_{wl+1}^n) = 0\right\}. \label{eq:PPM}
\end{align}
That is, we require each generated vector $\x$ to contain exactly $l$ ones; in particular, each of the first $l$ intervals $[1:w], [w+1:2w], \ldots, [(l-1)w+1:lw]$ contains a single one, and the last interval $[wl+1:n]$ contains all zeros. The PPM output distribution $\ppmz$, induced by $\ppmx$ and the channel $\wzx$, takes the form 
\begin{align}
    \ppmz(\z)  &\triangleq \sum_{\x\in \mathcal{X}^n} \ppmx(\x) \wzxn(\z|\x). \label{eq:ppmz}
\end{align}

\subsection{Analysis of the average probability of error} \label{sec:average}
We first consider the \emph{expected} average probability of error (over the codebook generation) for each sub-code $\C_k$, where $k \in \mathcal{K}$.
By applying a modified RCU bound that is adapted to the decoding metric $\mathsf{q}$, one can show that for every  $\C_k$, 
\begin{align}
    \mathbb{E}\left(P^{\mathrm{(avg)}}_{\mathrm{err}}(\C_k) \right)&= \mathbb{E}\left[\PP\left( \bigcup_{m'\in \mathcal{M} \setminus m} \left\{\mathsf{q}^n(\X(m'),\Y) \ge \mathsf{q}^n(\X(m),\Y) \right\} \bigg| \X(m),\Y \right)   \right]   \\
    &\le \mathbb{E}\left[ \min\left\{ 1, (|\mathcal{M}|-1) \times  \PP\left( \mathsf{q}^n(\widetilde{\X},\Y) \ge \mathsf{q}^n(\X,\Y) \bigg| \X,\Y \right) \right\} \right] \\
    &\le \mathbb{E}\left[ \min\left\{ 1, |\mathcal{M}| \times  \mathbb{E}\left[ \left(\frac{\mathsf{q}^n(\widetilde{\X},\Y)}{\mathsf{q}^n(\X,\Y)} \right)^s \bigg| \X,\Y \right] \right\} \right] \\
    &\le \PP\left[\log|\mathcal{M}| + \log \mathbb{E} \left[\left(\frac{\mathsf{q}^n(\widetilde{\X},\Y)}{\mathsf{q}^n(\X,\Y)} \right)^s \bigg| \X,\Y\right] \ge \log \lambda \right] + \lambda  \label{eq:error}
\end{align}
for any $s \ge 0$ and $\lambda > 0$, where $(\X,\Y,\widetilde{\X}) \sim \ppmx(\x) \cdot \wyxn(\y|\x) \cdot \ppmx(\widetilde{\x})$. Note that 
\begin{align}
&-\log \mathbb{E} \left[\left(\frac{\mathsf{q}^n(\widetilde{\X},\Y)}{\mathsf{q}^n(\X,\Y)} \right)^s \bigg| \X,\Y\right]  \label{eq:xiao} \\
&=-\log \sum_{\widetilde{\x}} \ppmx(\widetilde{\x}) \left(\frac{\mathsf{q}^n(\widetilde{\x},\Y)}{\mathsf{q}^n(\X,\Y)} \right)^s \\
&= -\log \sum_{\widetilde{\x}} \left(\prod_{i=1}^l \pwx(\widetilde{\ux}^{(i)}) \cdot \mathbbm{1}\left\{\text{wt}_{\mathrm{H}}(\widetilde{x}_{wl+1}^n) = 0\right\} \right) \left[\frac{\prod_{i=1}^l \mathsf{q}^w(\widetilde{\ux}^{(i)},\underline{\Y}^{(i)})^s \cdot \mathsf{q}^{n-wl}(\widetilde{x}_{wl+1}^n, Y_{wl+1}^n)^s }{\prod_{i=1}^l \mathsf{q}^w(\underline{\X}^{(i)},\underline{\Y}^{(i)})^s \cdot  \mathsf{q}^{n-wl}(X_{wl+1}^n, Y_{wl+1}^n)^s} \right] \\
&= - \log \left[\left(\prod_{i=1}^l \sum_{\widetilde{\ux}^{(i)}} \pwx(\widetilde{\ux}^{(i)}) \frac{\mathsf{q}^w(\widetilde{\ux}^{(i)},\underline{\Y}^{(i)})^s}{\mathsf{q}^w(\underline{\X}^{(i)},\underline{\Y}^{(i)})^s} \right) \times \left(\sum_{\widetilde{x}_{wl+1}^n}\mathbbm{1}\left\{\text{wt}_{\mathrm{H}}(\widetilde{x}_{wl+1}^n) = 0\right\} \frac{\mathsf{q}^{n-wl}(\widetilde{x}_{wl+1}^n, Y_{wl+1}^n)^s}{\mathsf{q}^{n-wl}(X_{wl+1}^n, Y_{wl+1}^n)^s} \right) \right] \\
&= \sum_{i=1}^l -\log \sum_{\widetilde{\ux}^{(i)}} \pwx(\widetilde{\ux}^{(i)}) \frac{\mathsf{q}^w(\widetilde{\ux}^{(i)},\underline{\Y}^{(i)})^s}{\mathsf{q}^w(\underline{\X}^{(i)},\underline{\Y}^{(i)})^s} - \log \frac{\mathsf{q}^{n-wl}(0^{n-wl}, Y_{wl+1}^n)^s}{\mathsf{q}^{n-wl}(X_{wl+1}^n, Y_{wl+1}^n)^s}. \label{eq:13}
\end{align}
The second term in~\eqref{eq:13} always equals zero since $X_{wl+1}^n = 0^{n-wl}$ with probability one (due to the property of the PPM distribution $\ppmx$). Thus, it suffices to focus on the first $l$ terms 
\begin{align}
    S_i \triangleq -\log \sum_{\widetilde{\ux}^{(i)}} \pwx(\widetilde{\ux}^{(i)}) \frac{\mathsf{q}^w(\widetilde{\ux}^{(i)},\underline{\Y}^{(i)})^s}{\mathsf{q}^w(\underline{\X}^{(i)},\underline{\Y}^{(i)})^s}. 
\end{align}

\begin{lemma} \label{lemma:expectation}
There exist two constants $\underline{B}$ and $\bar{B}$ such that $\underline{B} \le S_i \le \bar{B}$ (i.e., $S_i$ is bounded), and the expectation of $S_i$ satisfies
\begin{align}
    \mathbb{E}(S_i) \ge \mathbb{E}_{P_1}\left(\log\frac{\mathsf{q}(1,Y)^s}{\mathsf{q}(0,Y)^s}\right) - \log\mathbb{E}_{P_0}\left(\frac{\mathsf{q}(1,Y)^s}{\mathsf{q}(0,Y)^s} \right) - \frac{C_0}{w}
\end{align}
for some constant $C_0>0$.   
\end{lemma}

We prove Lemma~\ref{lemma:expectation} in Appendix~\ref{appendix:expectation}. Note that $\mathbb{E}(S_i)$ is finite since the decoding metric $\mathsf{q}$ is only allowed to take on positive values.
By applying Hoeffding's inequality and choosing $\lambda = \exp(-n^{1/4})$, one can bound the first term in~\eqref{eq:error} from above as 
\begin{align}
\PP \left[\sum_{i=1}^l S_i \le \log |\mathcal{M}| - \log \lambda  \right] &= \PP \left[\frac{\sum_{i=1}^l S_i}{l} - \mathbb{E}(S_i) \le \left(\sqrt{1+ \frac{n^{-1/3}}{2\delta - n^{-1/3}}} - 1 \right) \mathbb{E}(S_i) - \frac{\eta_1\sqrt{n}}{l} + \frac{n^{1/4}}{l} \right] \\
&\le 2 \exp \left(-\Theta(\sqrt{\delta n}\cdot \eta(\eta_1)^2) \right),
\end{align}
where $\eta(\eta_1) \to 0$ as $\eta_1 \to 0$. Therefore, for each sub-code $\C_k$, the expected average probability of error satisfies 
\begin{align}
    \mathbb{E}\left(P^{\mathrm{(avg)}}_{\mathrm{err}}(\C_k) \right) \le 2 \exp \left(-\Theta(\sqrt{\delta n}\cdot \eta(\eta_1)^2) \right) + \exp(-n^{1/4}) \le 2\exp(-n^{1/4}) 
\end{align}
for sufficiently large $n$.

\subsection{Analysis of the randomly generated code $\C$} \label{sec:existence}
We first introduce a result showing that, with a positive probability, the randomly generated code $\C$ simultaneously satisfies the two properties mentioned in the beginning of this section.

\begin{lemma}[Adapted from Lemma 4 in~\cite{tahmasbi2018first}]  \label{lemma:subset}
For every $\lambda_1, \lambda_2, \gamma > 0$, we have that if
\begin{align}
    1 - \frac{1}{n} > \exp \left\{-|\mathcal{M}| \left(\frac{2 \lambda_1 \lambda_2^2}{\log^2\left(\frac{\lambda_1 |\mathcal{M}||\mathcal{K}||\mathcal{Z}^n|}{([\ppmz]^{\min})^2} \right)} - H_b(\lambda_1) \right) \right\}, \label{eq:huo1}
\end{align}
where $H_b(\cdot)$ is the binary entropy function, then with a positive probability, the following two events occur simultaneously:
\begin{itemize}
\item Event $\mathcal{E}_1$: For every sub-codes $\{\C_k\}_{k \in \mathcal{K}}$, the average probabilities of error satisfies \begin{align}
    P^{\mathrm{(avg)}}_{\mathrm{err}}(\C_k) \le 2n\exp(-n^{1/4}). \label{eq:34}
\end{align}
\item Event $\mathcal{E}_2$: For every subset of the codebook $\mathcal{I} \subset \mathcal{M} \times \mathcal{K}$ with cardinality $|\mathcal{I}| = \lambda_1 |\mathcal{M}||\mathcal{K}|$, its induced output distribution $\widehat{Q}^n_{\mathcal{I}}$, which takes the form $\widehat{Q}^n_{\mathcal{I}}(\z) = \frac{1}{|\mathcal{I}|} \sum_{(m,k) \in \mathcal{I}} \wzxn(\z|\x_{m,k})$, satisfies 
\begin{align}
    \DD \left(\widehat{Q}^n_{\mathcal{I}} \Vert \ppmz \right) \le \log\left(1 + \frac{1}{[\ppmz]^{\min}} \right) \times \PP_{\ppmx \wzxn} \left( \log \frac{\wzxn(\Z|\X)}{\ppmz(\Z)} > \gamma \right) + \frac{\exp(\gamma)}{\lambda_1 |\mathcal{M}||\mathcal{K}|} + \lambda_2. \label{eq:huo2}
\end{align}
\end{itemize}
\end{lemma}

The proof of Lemma~\ref{lemma:subset} is adapted from that for~\cite[Lemma 4]{tahmasbi2018first} with appropriate modifications, and the details are provided in Appendix~\ref{appendix:b}.
In the following, we evaluate~\eqref{eq:huo1} and~\eqref{eq:huo2} by setting $\lambda_1 = 1 - n^{-11}$, $\lambda_2 = n^{-4}$, and $\gamma = \log(|\mathcal{M}||\mathcal{K}|/n^4)$. First note that $[\ppmz]^{\min} \ge (\min\{[Q_0]^{\min},[Q_1]^{\min}\})^n$, and the RHS of~\eqref{eq:huo1} satisfies
\begin{align}
\exp \left\{-|\mathcal{M}| \left(\frac{2 \lambda_1 \lambda_2^2}{\log^2\left(\frac{\lambda_1 |\mathcal{M}||\mathcal{K}||\mathcal{Z}^n|}{([\ppmz]^{\min})^2} \right)} - H_b(\lambda_1) \right) \right\} \le \exp\big\{-|\mathcal{M}|  \cdot \left( C_1 n^{-10} - C_2 n^{-11}  \log n \right) \big\}   
\end{align}
for some constants $C_1, C_2 > 0$. Therefore, the condition in~\eqref{eq:huo1} holds for sufficiently large $n$, and there must exist a code $\C = \{\C_k\}_{k \in \mathcal{K}}$, with each $|\C_k| = |\mathcal{M}|$, satisfying~\eqref{eq:34} and~\eqref{eq:huo2} simultaneously.

We next evaluate~\eqref{eq:huo2}. According to~\cite[Lemma 7]{tahmasbi2018first}, we have 
\begin{align}
    \PP_{\ppmx \wzxn} \left( \log \frac{\wzxn(\Z|\X)}{\ppmz(\Z)} > \gamma \right) &\le \exp\left\{ - \frac{[\gamma-l\DD(Q_1 \Vert Q_0)]^2}{l\cdot (\log w)^2} \right\} \\
    &\le  \exp\left\{ - \frac{[t\sqrt{n}\DD(Q_1\Vert Q_0) + \eta_2 \sqrt{n} - 4\log n-l\DD(Q_1 \Vert Q_0)]^2}{l\cdot (\log w)^2} \right\}, \label{eq:lin}
\end{align}
where~\eqref{eq:lin} is due to the choices of $\gamma$ and $|\mathcal{M}||\mathcal{K}|$. 
For sufficiently large $n$, the first term in the RHS of~\eqref{eq:huo2} satisfies
\begin{align}
 &\log\left(1 + \frac{1}{[\ppmz]^{\min}} \right) \times \PP_{\ppmx \wzxn} \left( \log \frac{\wzxn(\Z|\X)}{\ppmz(\Z)} > \gamma \right) \\
 &\le \left[(n+1) \log \frac{1}{\min\{[Q_0]^{\min},[Q_1]^{\min}\}}\right] \times \exp\left\{ - \frac{[t\sqrt{n}\DD(Q_1\Vert Q_0) + \eta_2 \sqrt{n} - 4\log n-l\DD(Q_1 \Vert Q_0)]^2}{l\cdot (\log w)^2} \right\} \\
 &\le \exp\left(- C_3\sqrt{n} \right),
\end{align}
for some constant $C_3 > 0$. Thus, for every subset of the codebook $|\mathcal{I}| \subset \mathcal{M} \times \mathcal{K}$ with cardinality $\mathcal{I} = (1-n^{-11})|\mathcal{M}||\mathcal{K}|$, its induced output distribution $\widehat{Q}^n_{\mathcal{I}}$ satisfies 
\begin{align}
    \DD \left(\widehat{Q}^n_{\mathcal{I}} \Vert \ppmz \right) \le \exp\left(- C_3\sqrt{n} \right) + \frac{\exp(\log(|\mathcal{M}||\mathcal{K}|/n^4))}{(1-n^{-11}) |\mathcal{M}||\mathcal{K}|} + n^{-4} \le 3n^{-4}.   \label{eq:42}
\end{align}

\subsection{Analysis of the expurgated code $\widetilde{\C}$} \label{sec:expurgate}
As described in the beginning of this section, we construct a new code $\widetilde{\C} = \{\widetilde{\C}_k\}_{k\in\mathcal{K}}$ based on the original code $\C = \{\C_k\}_{k\in\mathcal{K}}$ that satisfies~\eqref{eq:34} and~\eqref{eq:huo2}, by expurgating $n^{-11}|\mathcal{M}|$ codewords that have the highest probabilities of error from each of the sub-code $\{\C_k\}_{k\in\mathcal{K}}$. Since $P^{\mathrm{(avg)}}_{\mathrm{err}}(\C_k) \le 2n\exp(-n^{1/4})$ for every $k \in \mathcal{K}$, by Markov's inequality we have that
\begin{align}
  P^{\mathrm{(max)}}_{\mathrm{err}}(\widetilde{\C}_k) \le \frac{1}{n^{-11}} \cdot 2n\exp(-n^{1/4}) = 2n^{12} \exp(-n^{1/4})
\end{align}
for every $k \in \mathcal{K}$. This implies that 
\begin{align}
\lim_{n \to \infty} P^{\mathrm{(max)}}_{\mathrm{err}}(\widetilde{\C}) = \lim_{n \to \infty} \max_{k \in \mathcal{K}}P^{\mathrm{(max)}}_{\mathrm{err}}(\widetilde{\C}_k) = 0.    
\end{align}
Since the new code $\widetilde{\C}$ is a subset of the original code $\C$ and satisfies $|\widetilde{\C}| = (1-n^{-11})|\mathcal{M}||\mathcal{K}|$, its induced output distribution $\widehat{Q}^n_{\widetilde{\C}}$ satisfies~\eqref{eq:42}, i.e., $\DD \left(\widehat{Q}^n_{\widetilde{\C}} \Vert \ppmz \right) \le 3 n^{-4}$.
Following the analysis in~\cite[Lemmas 4 and 5]{zhang2021covert}, one can further show that the choice of $l$ implies 
\begin{align}
    &\DD \left(\ppmz \Vert Q_0^{\otimes n} \right) \le \delta - \frac{1}{3} n^{-1/3}, \quad \mathrm{and} \\
    &\sum_{\z} \left(\widehat{Q}^n_{\widetilde{\C}}(\z) - \ppmz(\z) \right) \log \frac{\ppmz(\z)}{Q_0^{\otimes n}(\z)} \le 2n \left(\log\frac{1}{[Q_0]^{\min}} \right) \sqrt{\DD \left(\widehat{Q}^n_{\widetilde{\C}} \Vert \ppmz \right)}.
\end{align}
Therefore, the KL divergence between $\widehat{Q}^n_{\widetilde{\C}}$ and $Q_0^{\otimes n}$ can be bounded from above as 
\begin{align}
    \DD\left(\widehat{Q}^n_{\widetilde{\C}} \Vert Q_0^{\otimes n} \right) &= \DD \left(\widehat{Q}^n_{\widetilde{\C}} \Vert \ppmz \right) + \DD \left(\ppmz \Vert Q_0^{\otimes n} \right) + \sum_{\z} \left(\widehat{Q}^n_{\widetilde{\C}}(\z) - \ppmz(\z) \right) \log \frac{\ppmz(\z)}{Q_0^{\otimes n}(\z)} \le \delta,
\end{align}
where the last inequality is valid for sufficiently large $n$. Finally, we note that
\begin{align}
        \liminf_{n \to \infty} \frac{\log |\widetilde{\C}_k|}{\sqrt{n}} = \liminf_{n \to \infty} \frac{\log|\C_k|}{\sqrt{n}} = \liminf_{n \to \infty} \frac{\log|\mathcal{M}|}{\sqrt{n}} = \underline{\mathsf{R}}_{\mathsf{q},\delta} -\eta_1,
\end{align}
and the proof is completed by taking $\eta_1 \to 0^+$.

\section{Proof of Theorem~\ref{thm:converse} (Upper Bound)} \label{sec:converse}

Similar to the definition of the maximal joint conditional distributions in Definition~\ref{def:maximal}, below we introduce the definition of the \emph{maximal joint conditional type}.

\begin{definition}[Maximal joint conditional type] \label{def:maximal2}
A joint conditional type $T_{Y\widehat{Y}|X}$ is called a maximal joint conditional type if  
\begin{align}
    T_{Y\widehat{Y}|X}(y,\yh|x) = 0, \quad \mbox{for all} \ (y, \yh) \in \mathcal{Y} \times \mathcal{Y} \ \mbox{and} \ x \notin \Sq(y,\yh). \label{eq:maximal}
\end{align}
\end{definition}

Suppose there exists a $\delta$-covert code $\C$ containing $|\K|$ sub-codes $\{\C_k\}_{k \in \K}$ of size $|\M|$, where $\frac{\log |\M|}{\sqrt{n}} = \Rbar + \sigma$ for some $\sigma > 0$. Since the code $\C$ satisfies the covertness constraint $\DD(\widehat{Q}^n_{\C} \Vert Q_0^{\otimes n}) \le \delta$, by~\cite[Eqn. (234)]{tahmasbi2018first} we know that there exists a subset of codewords $\D \subseteq \C$ and a vanishing sequence $\gamma_n$ such that 
\begin{align}
\mathrm{(i)} \ |\D| \ge \frac{|\M||\K|}{n} \quad \mathrm{and} \quad   \mathrm{(ii)} \ \max_{\x \in \D}  \frac{\mathrm{wt_H}(\x)}{\sqrt{n}} \le (1 + \gamma_n) t_{\delta}.
\end{align}
For each sub-code $\C_i$, we define $\C_i^\D \triangleq \C_i \cap \D$ as the intersection between $\C_i$ and the subset $\D$. By the Pigeonhole principle, there must exist a subset $\C_i^\D$ such that $|\C_i^\D| \ge |\M|/n$ and $\max_{\x \in \C_i^\D}\mathrm{wt_H}(\x) \le (1+\gamma_n)t_{\delta}\sqrt{n}$. By applying the Pigeonhole principle again to $\C_i^\D$, one can obtain a constant composition code $\C_i^{cc} \subseteq \C_i^\D$ such that 
\begin{align}
\mathrm{(i)} \ |\C_i^{cc}| \ge \frac{|\M|}{(1+\gamma_n)t_{\delta} n^{3/2}} \quad \mathrm{and} \quad   \mathrm{(ii)} \ \forall \x \in \C_i^{cc}, \ T_{\x}(1) = 1 - T_{\x}(0)  = \frac{t'_n}{\sqrt{n}} \ \mbox{for some } 0\le t'_n \le (1+\gamma_n)t_{\delta}. \label{eq:cc}
\end{align}
We denote the type of the codewords in $\C_i^{cc}$ by $P^{cc}$, i.e., $T_{\x} = P^{cc}$ for all $\x \in \C_i^{cc}$.
Note that $P_{\mathrm{err}}^{\mathrm{(max)}}(\C_i) \ge P_{\mathrm{err}}^{\mathrm{(max)}}(\C_i^{cc})$ by the definition of maximum probability of error, thus it suffices to  show that $P_{\mathrm{err}}^{\mathrm{(max)}}(\C_i^{cc})$ is bounded away from zero.

If the weight parameter $t'_n$ of codewords in $\C_i^{cc}$ satisfies $t'_n \in o(\frac{1}{\log n})$, then there must exist two identical codewords $\x, \x' \in \C_i^{cc}$ (since the size of the type class $\binom{n}{t'_n\sqrt{n}}$ is even smaller than the number of codewords in $\C_i^{cc}$), thus  $P_{\mathrm{err}}(\C_i^{cc})$ must be bounded away from zero. Therefore, in the following, we only consider the scenario in which $t'_n \notin o(\frac{1}{\log n})$, which is equivalent to saying that $\limsup_{n \to \infty} \frac{t'_n}{1/\log n} > 0$. This implies that there exists a subsequence of blocklengths $\{n_k\}_{k=1}^{\infty}$ such that $t'_{n_k}\log n_k > \varepsilon$ for some $\varepsilon > 0$. By passing to a subsequence (as above) if necessary, we can assume that $t'_n \log n$ is a convergent sequence and its limit is greater than zero. In the following, we also abbreviate $t'_n$ as $t'$ for notational convenience.      

 We denote the set of conditional types $V \in \mathcal{V}_{P^{cc}}(\mathcal{Y}|\mathcal{X})$ that are close to the channel $\wyx$ as 
\begin{align}
\mathscr{V}_{P^{cc}} \triangleq \left\{V \in \mathcal{V}_{P^{cc}}(\mathcal{Y}|\mathcal{X}): \forall y \in \mathcal{Y}, \ \left|V(y|1) - P_1(y) \right| \le [P_1]^{\max} n^{-1/8}, \ \left|V(y|0) - P_0(y) \right| \le [P_0]^{\max} \sqrt{\frac{\log n}{n}} \right\},  
\end{align}
and recall that $P_0 = W_{Y|X=0}$ and $P_1 = W_{Y|X=1}$.
\begin{lemma} \label{lemma:typical}
For any codeword $\x \in \C_i^{cc}$, with probability at least $1 - 2n^{-\frac{1}{3}[P_0]^{\min}}$, Bob's received sequence $\Y$ satisfies $T_{\Y|\x} \in \mathscr{V}_{P^{cc}}$.
\end{lemma}
Lemma~\ref{lemma:typical} can be proved via standard concentration inequalities, and we defer the detailed proof to Appendix~\ref{appendix:typical}.
We now consider the joint conditional distribution $\Pj^\ast$ that minimizes~\eqref{eq:minimizer}, and note that its marginal distributions must satisfy $P^{\ast}_{\widehat{Y}|X=1} \ll P^{\ast}_{\widehat{Y}|X=0}$, since $\DD\big(P^{\ast}_{\widehat{Y}|X=1} \Vert P^{\ast}_{\widehat{Y}|X=0}\big)$ would be infinite otherwise. We denote the set of conditional types $V \in \mathcal{V}_{P_X^{cc}}(\mathcal{Y}|\mathcal{X})$ that are close to $P^{\ast}_{\widehat{Y}|X}$   as 
\begin{align}
\widehat{\mathscr{V}}_{P^{cc}} \triangleq \Bigg\{\Vh \in \mathcal{V}_{P_X^{cc}}(\mathcal{Y}|\mathcal{X}): \forall \yh \in \mathcal{Y}, \ \left|\Vh(\yh|1) - P^{\ast}_{\widehat{Y}|X}(\yh|1) \right| \le \frac{[P_1]^{\max}|\mathcal{Y}|}{n^{1/8}}+ \frac{|\mathcal{Y}|^2+|\mathcal{Y}|}{t'\sqrt{n}}, \notag  \\
\ \left|\Vh(\yh|0) - P^{\ast}_{\widehat{Y}|X}(\yh|0) \right| \le \frac{[P_0]^{\max}|\mathcal{Y}| \sqrt{\log n}}{\sqrt{n}} + \frac{2|\mathcal{Y}|^2+|2\mathcal{Y}|}{n} \Bigg\}.\label{eq:type2}
\end{align}

Lemma~\ref{lemma:a} below shows that regardless of the conditional type $V \in \mathscr{V}_{P^{cc}}$ that Bob's received sequence $\Y$ falls into, there always exists a maximal joint conditional type $\bar{V}_{Y\widehat{Y}|X}$ such that its marginal conditional types $\bar{V}_{Y|X} = V$ and $\bar{V}_{\widehat{Y}|X} \in \widehat{\mathscr{V}}_{P^{cc}}$, thus $\bar{V}_{Y\widehat{Y}|X}$ is also close to $\Pj^\ast$.

\begin{lemma} \label{lemma:a}
    For any conditional type $V' \in \mathscr{V}_{P^{cc}}$ satisfying that $V'(y|x) = 0$ if $\wyx(y|x) = 0$, there exists a maximal joint conditional type $\bar{V}_{Y\widehat{Y}|X}$ such that its marginal conditional types $\bar{V}_{Y|X} = V'$ and $\bar{V}_{\widehat{Y}|X} \in \widehat{\mathscr{V}}_{P^{cc}}$.
\end{lemma}

\begin{proof}
See Appendix~\ref{appendix:a}.
\end{proof}

Next, we present a key lemma (Lemma~\ref{lemma:key}) stating that if the size of constant composition code $\C_i^{cc}$ were too large (as quantified in~\eqref{eq:condition}), then there would exist a codeword $\x(m)$ having a non-vanishing probability of error (under the decoding metric $\mathsf{q}$), thus $P_{\mathrm{err}}^{\mathrm{(max)}}(\C_i^{cc})$ would be non-vanishing. Lemma~\ref{lemma:key} can be viewed as a strengthened version of~\cite[Theorem 3]{kangarshahi2021single}, which was developed by Kangarshahi and Guillén i Fàbregas for the classical mismatched decoding problem.


\begin{lemma} \label{lemma:key}
If the constant composition code $\C_i^{cc}$ satisfies that for some integer $a \ge 2$, for all $\x \in \T_{P^{cc}}$ and all $\Vh \in \widehat{\mathscr{V}}_{P^{cc}}$,
    \begin{align}
        |\C_i^{cc}|\cdot \frac{|\T_{\Vh}(\x)|}{|\widehat{\mathscr{V}}_{P^{cc}}| \cdot |\T_{P^{cc}\Vh}|} \ge a^2,  \label{eq:condition}
    \end{align}
    where $P^{cc}\Vh$ is the marginal output type induced by $P^{cc}$ and $\Vh$, 
    then there exists a codeword $\x(m) \in \C_i^{cc}$ such that
    \begin{align}
        \PP\left(\widehat{M} \ne m|M = m, T_{\Y|\x(m)} \in \mathscr{V}_{P^{cc}} \right) \ge 1 - \frac{2}{a+1}.
    \end{align}
\end{lemma}
\begin{proof}
See Appendix~\ref{appendix:key}.
\end{proof}

Roughly speaking, the term $|\T_{P^{cc}\Vh}|/|\T_{\Vh}(\x)|$ in~\eqref{eq:condition} corresponds to the mutual information of the conditional type $\widehat{V} \in \widehat{\mathscr{V}}_{P^{cc}}$, thus one can interpret~\eqref{eq:condition} as the condition that the rate exceeds the mutual information. We remark that the key difference between Lemma~\ref{lemma:key} and~\cite[Theorem 3]{kangarshahi2021single} lies in the condition~\eqref{eq:condition}, where~\cite[Theorem 3]{kangarshahi2021single}  instead requires that 
\begin{align}
    |\C_i^{cc}|\cdot \frac{\min_{\Vh \in \widehat{\mathscr{V}}_{P^{cc}}} |\T_{\Vh}(\x)|}{|\widehat{\mathscr{V}}_{P^{cc}}| \cdot \left(\max_{\Vh \in \widehat{\mathscr{V}}_{P^{cc}}} |\T_{P^{cc}\Vh}| \right)} \ge a^2.  \label{eq:condition2}
\end{align}
In standard (non-covert) communication, the effect of the minimization and maximization over $\Vh \in \widehat{\mathscr{V}}_{P^{cc}}$ is essentially negligible, thus $\max_{\Vh \in \widehat{\mathscr{V}}_{P^{cc}}} |\T_{P^{cc}\Vh}|/\min_{\Vh \in \widehat{\mathscr{V}}_{P^{cc}}} |\T_{\Vh}(\x)|$ can approximately be interpreted as the mutual information of \emph{any} conditional type $\Vh \in \widehat{\mathscr{V}}_{P^{cc}}$.
However, due to the stringent constraint on the Hamming weight of $\x$ in covert communication, the two quantities $|\T_{\Vh}(\x)|$ and $|\T_{P^{cc}\Vh}|$ have the same first-order term and only differ in the second-order terms. This further implies that the effect of the minimization and maximization over $\Vh \in \widehat{\mathscr{V}}_{P^{cc}}$ is \emph{not} negligible, and in fact, the constant composition code $\C_i^{cc}$ does not satisfy~\eqref{eq:condition2} for any $a \ge 2$. To overcome this challenge, we provide a strengthened result in Lemma~\ref{lemma:key}, where the condition~\eqref{eq:condition} is more relaxed compared to the condition~\eqref{eq:condition2} in~\cite[Theorem 3]{kangarshahi2021single}.

Next, we show that the constant composition code $\C_i^{cc}$ satisfies Lemma~\ref{lemma:key} for $a = \exp((\sigma/4)n)$. For $\Vh \in \widehat{\mathscr{V}}_{P^{cc}}$, we denote its output types corresponding to $X= 0$ and $X=1$ respectively as $\Vh_0$ and $\Vh_1$, thus one can show that for any $\x \in \T_{P^{cc}}$, 
\begin{align}
    \left|\T_{\Vh}(\x)\right| &\ge (n+1)^{-|\mathcal{X}||\mathcal{Y}|} \exp\left\{n \left(\frac{t'}{\sqrt{n}}H(\Vh_1) + \left(1- \frac{t'}{\sqrt{n}} \right)H(\Vh_0) \right) \right\}, \quad \mathrm{and}  \label{eq:xi1}\\
    \left|\T_{P^{cc}\Vh}\right| &\le \exp\left\{n H\left(\frac{t'}{\sqrt{n}}\Vh_1 + \left(1 - \frac{t'}{\sqrt{n}} \right) \Vh_0 \right) \right\}. 
\end{align}
By Taylor expansions and some simple algebras (as detailed in Appendix~\ref{appendix:calculation1}), one can show that  
\begin{align}
&H\left(\frac{t'}{\sqrt{n}}\Vh_1 + \left(1 - \frac{t'}{\sqrt{n}} \right) \Vh_0 \right) - \frac{t'}{\sqrt{n}}H(\Vh_1) - \left(1- \frac{t'}{\sqrt{n}} \right)H(\Vh_0)  \le \frac{t'}{\sqrt{n}} \DD\left(\Vh_1 \big\Vert \Vh_0 \right) + \frac{C_4\cdot (t')^2}{n} \label{eq:c}
\end{align}
for some constant $C_4 > 0$. Since each $\Vh \in \widehat{\mathscr{V}}_{P^{cc}}$ is close to $P^{\ast}_{\widehat{Y}|X}$, we show in Appendix~\ref{appendix:d} that there exists a vanishing sequence $\mu_n$ such that
\begin{align}
\DD\left(\Vh_1 \big\Vert \Vh_0 \right) \le \DD\left(P^{\ast}_{\widehat{Y}|X=1} \Vert P^{\ast}_{\widehat{Y}|X=0} \right) + \mu_n. \label{eq:xi2}
\end{align}
Thus, by combining~\eqref{eq:xi1}-\eqref{eq:xi2} and noting that $|\widehat{\mathscr{V}}_{P^{cc}}| \le (n+1)^{|\mathcal{X}||\mathcal{Y}|}$, we have 
\begin{align}
&|\C_i^{cc}|\cdot \frac{|\T_{\Vh}(\x)|}{|\widehat{\mathscr{V}}_{P^{cc}}| \cdot |\T_{P^{cc}\Vh}|} \notag \\
&\ge \frac{|\M|}{(1+\gamma_n)t_{\delta}n^{3/2}} \exp\left\{-t'\sqrt{n} \left[\DD\left(P^{\ast}_{\widehat{Y}|X=1} \Vert P^{\ast}_{\widehat{Y}|X=0}\right) + \mu_n \right] - C_4\cdot (t')^2 \right\} (n+1)^{-2|\mathcal{X}||\mathcal{Y}|}  \\
&= \exp\left\{ \sqrt{n} (\Rbar + \sigma) - t'\sqrt{n}\left[\DD\left(P^{\ast}_{\widehat{Y}|X=1} \Vert P^{\ast}_{\widehat{Y}|X=0}\right) + \mu_n \right] - C_4\cdot (t')^2 \right\} \frac{(n+1)^{-2|\mathcal{X}||\mathcal{Y}|}}{(1+\gamma_n)tn^{3/2}} \\
& \ge \exp\left\{\sqrt{n}\left[t_{\delta}\DD\left(P^{\ast}_{\widehat{Y}|X=1} \Vert P^{\ast}_{\widehat{Y}|X=0}\right) + \sigma\right] - (1+\gamma_n)t_{\delta}\sqrt{n}\left[\DD\left(P^{\ast}_{\widehat{Y}|X=1} \Vert P^{\ast}_{\widehat{Y}|X=0}\right) + \mu_n\right] - C_4\cdot (t')^2 \right\} \frac{(n+1)^{-2|\mathcal{X}||\mathcal{Y}|}}{(1+\gamma_n)t_{\delta}n^{3/2}} \\
& \ge \exp\{(\sigma/2)\sqrt{n} \},
\end{align}
for sufficiently large $n$. Thus, the constant composition code $\C_i^{cc}$ satisfies Lemma~\ref{lemma:key} for $a = \exp((\sigma/4)n)$, and there exists a codeword $\x(m)$ such that 
    \begin{align}
        \PP\left(\widehat{M} \ne m|M = m, T_{\y|\x(m)} \in \mathscr{V}_{P^{cc}} \right) \ge 1 - \frac{2}{\exp((\sigma/4)n)+1}.
    \end{align}
Finally, the maximum probability of error $P_{\mathrm{err}}^{\mathrm{(max)}}(\C_i^{cc})$ can be bounded from below as 
\begin{align}
    P_{\mathrm{err}}^{\mathrm{(max)}}(\C_i^{cc}) &\ge \PP(\widehat{M} \ne m|M = m) \\
    &\ge \PP\left(\widehat{M} \ne m|M = m, T_{\Y|\x(m)} \in \mathscr{V}_{P^{cc}} \right) \cdot \PP\left(T_{\Y|\x(m)} \in \mathscr{V}_{P^{cc}}|M = m \right) \\
    &\ge \left(1 - \frac{2}{\exp((\sigma/4)n)+1} \right) \times \left(1 - 2n^{-\frac{1}{3}[P_0]^{\min}} \right),
\end{align}
which tends to one as $n$ tends to infinity. This completes the proof of the upper bound.

\appendices

\section{Proof of Theorem~\ref{thm:binary}} \label{appendix:binary}

To simplify the expression of the achievable rate $\underline{\mathsf{R}}_{\mathsf{q},\delta}$, we first define 
$f(s) \triangleq  \mathbb{E}_{P_1}\!\!\left[\log\frac{\mathsf{q}(1,Y)^s}{\mathsf{q}(0,Y)^s}\right] - \log \mathbb{E}_{P_0}\!\!\left[\frac{\mathsf{q}(1,Y)^s}{\mathsf{q}(0,Y)^s} \right]$,
and note that $\underline{\mathsf{R}}_{\mathsf{q},\delta} = t_{\delta} \cdot \left(\sup_{s \ge 0} f(s)\right)$. The derivative of $f(s)$ is 
\begin{align}
    \frac{\mathrm{d} f(s)}{\mathrm{d} s} = \left(\log \frac{\mathsf{q}(0,0)\mathsf{q}(1,1)}{\mathsf{q}(0,1)\mathsf{q}(1,0)} \right) \times \left(\frac{1}{1 + \frac{1-P_0(0)}{P_0(0)} \frac{\mathsf{q}(0,0)^s \mathsf{q}(1,1)^s}{\mathsf{q}(0,1)^s \mathsf{q}(1,0)^s}}  - P_1(0)\right). \label{eq:derivative}
\end{align}
\paragraph{When $\mathsf{q}(0,0)   \mathsf{q}(1,1) > \mathsf{q}(0,1)  \mathsf{q}(1,0)$} One can check that the  derivative 
    \begin{align}
    \begin{cases}
    \frac{\mathrm{d} f(s)}{\mathrm{d} s} \ge 0, &\mbox{when } \ 0 \le s < s_0, \\
    \frac{\mathrm{d} f(s)}{\mathrm{d} s} = 0, &\mbox{when } \ s = s_0, \\
    \frac{\mathrm{d} f(s)}{\mathrm{d} s} \le 0, &\mbox{when } \ s > s_0,
    \end{cases}
    \quad \mbox{where} \ s_0 = \frac{\log \frac{P_0(0)P_1(1)}{P_0(1)P_1(0)}}{\log\frac{\mathsf{q}(0,0)\mathsf{q}(1,1)}{\mathsf{q}(0,1)\mathsf{q}(1,0)}}.
    \end{align}
Note that $s_0$ is non-negative since the numerator $\log \frac{P_0(0)P_1(1)}{P_0(1)P_1(0)} \ge 0$ when the assumption $P_0(0)+P_1(1)\ge P_0(1)+P_1(0)$ holds. Thus, $f(s)$ achieves its maximum when $s = s_0$, and
\begin{align}
\underline{\mathsf{R}}_{\mathsf{q},\delta} =  t_{\delta} \cdot f(s_0) =  t_{\delta} \cdot  \DD(P_1 \Vert P_0).
\end{align}
Next, we analyze the upper bound $\bar{\mathsf{R}}_{\mathsf{q},\delta}$ in Theorem~\ref{thm:converse}. By noting that $\Sq(0,0) = \Sq(1,1) = \{0,1\}$, $\Sq(0,1) = \{1\}$, $\Sq(1,0) = \{0\}$, one can verify that the set of distributions $\{P_{Y\widehat{Y}|X} \in \mathcal{M}_{\max}(\mathsf{q}): P_{Y|X} = \wyx\}$ can be characterized by
\begin{align*}
\begin{bmatrix}
P_{Y\widehat{Y}|X}(00|0) = P_0(0) & P_{Y\widehat{Y}|X}(01|0) = 0\\
P_{Y\widehat{Y}|X}(10|0) = r_1 & P_{Y\widehat{Y}|X}(11|0) = P_0(1)-r_1
\end{bmatrix} \ \mathrm{and} \ 
\begin{bmatrix}
P_{Y\widehat{Y}|X}(00|1) = P_1(0)-r_2 &  P_{Y\widehat{Y}|X}(01|1) = r_2\\
P_{Y\widehat{Y}|X}(10|1) = 0 & P_{Y\widehat{Y}|X}(11|1) = P_1(1)
\end{bmatrix}
\end{align*}
for $r_1 \in [0,P_0(0)]$ and $r_2 \in [0, P_1(0)]$. Thus, the marginal distribution $P_{\widehat{Y}|X}$ takes the form
\begin{align}
  P_{\widehat{Y}|X}  = \begin{bmatrix}
P_0(0) + r_1 & P_0(1) - r_1\\
P_1(0)-r_2 & P_1(1) + r_2
\end{bmatrix}, \label{eq:matrix}
\end{align}
in which the rows represent $X$ and the columns represent $\hat{Y}$.
Note that  
\begin{align}
    \DD(P_{\widehat{Y}|X=1} \Vert P_{\widehat{Y}|X=0}) &= (P_1(0)-r_2)\log \frac{P_1(0)-r_2}{P_0(0)+r_1} + (P_1(1)+r_2)\log \frac{P_1(1)+r_2}{P_0(1)-r_1} \label{eq:hei1} \\
    &=\DD(\mathrm{Bern}(P_1(0)-r_2) \Vert \mathrm{Bern}(P_0(0)+r_1)) \label{eq:hei2} \\
    &\ge \DD(P_1 \Vert P_0), \label{eq:hei3}
\end{align}
where~\eqref{eq:hei2} holds since $P_1(1)+r_2 = 1 - (P_1(0)-r_2)$ and $P_0(1)-r_1 = 1 - (P_0(0)+r_1)$, and~\eqref{eq:hei3} can be verified by calculating partial derivatives with respect to $r_1$ and $r_2$.  The equality in~\eqref{eq:hei3} is achieved when $r_1 = r_2 =0$. Thus, the upper bound $\bar{\mathsf{R}}_{\mathsf{q},\delta} = t_{\delta} \cdot \DD(P_1 \Vert P_0)$.   

\paragraph{When $\mathsf{q}(0,0)   \mathsf{q}(1,1) < \mathsf{q}(0,1)  \mathsf{q}(1,0)$} The second term in~\eqref{eq:derivative} is at least $P_0(0) - P_1(0)$ for all $s \ge 0$, and $P_0(0) - P_1(0)$ is non-negative when the assumption $P_0(0) + P_1(1) \ge P_0(1) + P_1(0)$ holds. Thus, $\frac{\mathrm{d} f(s)}{\mathrm{d} s} \le 0$ for all $s \ge 0$, and 
\begin{align}
\underline{\mathsf{R}}_{\mathsf{q},\delta} = t_{\delta} \cdot f(0) = 0.    
\end{align}
Next, we examine the upper bound $\bar{\mathsf{R}}_{\mathsf{q},\delta}$. By noting that $\Sq(0,0) = \Sq(1,1) = \{0,1\}$, $\Sq(0,1) = \{0\}$, $\Sq(1,0) = \{1\}$, one can check that the set $\{P_{Y\widehat{Y}|X} \in \mathcal{M}_{\max}(\mathsf{q}): P_{Y|X} = \wyx\}$ can be characterized by
\begin{align*}
\begin{bmatrix}
P_{Y\widehat{Y}|X}(00|0) = P_0(0)-r_1 & P_{Y\widehat{Y}|X}(01|0) = r_1\\
P_{Y\widehat{Y}|X}(10|0) = 0 & P_{Y\widehat{Y}|X}(11|0) = P_0(1)
\end{bmatrix} \ \mathrm{and} \ 
\begin{bmatrix}
P_{Y\widehat{Y}|X}(00|1) = P_1(0) &  P_{Y\widehat{Y}|X}(01|1) = 0\\
P_{Y\widehat{Y}|X}(10|1) = r_2 & P_{Y\widehat{Y}|X}(11|1) = P_1(1) - r_2
\end{bmatrix}
\end{align*}
for $r_1 \in [0,P_0(0)]$ and $r_2 \in [0, P_1(1)]$. We then have 
\begin{align}
  P_{\widehat{Y}|X}  = \begin{bmatrix}
P_0(0) - r_1 & P_0(1) + r_1\\
P_1(0)+r_2 & P_1(1) - r_2
\end{bmatrix}, \label{eq:matrix}
\end{align}
where rows represent $X$ and columns represent $\widehat{Y}$. Note that the assumption $P_0(0)+P_1(1)\ge P_0(1)+P_1(0)$ implies $P_0(0) \ge P_1(0)$, $P_1(1) \ge P_0(1)$, and $P_0(0)-P_1(0) = P_1(1)+P_0(1)$. 
Thus, it is valid to set $r_1 = r_2 = \frac{P_0(0)-P_1(0)}{2} = \frac{P_1(1) - P_0(1)}{2}$, in which case the two rows in~\eqref{eq:matrix} are identical,  yielding $\DD(P_{\widehat{Y}|X=1} \Vert P_{\widehat{Y}|X=0}) = 0$. Therefore, the upper bound $\bar{\mathsf{R}}_{\mathsf{q},\delta} = 0$. 

\paragraph{When $\mathsf{q}(0,0) \mathsf{q}(1,1) = \mathsf{q}(0,1)  \mathsf{q}(1,0)$} It is clear that $\frac{\mathrm{d} f(s)}{\mathrm{d} s} = 0$ for all $s \ge 0$, thus $\underline{\mathsf{R}}_{\mathsf{q},\delta} = t_{\delta} \cdot f(0) = 0$.    
To examine the upper bound $\bar{\mathsf{R}}_{\mathsf{q},\delta}$, we first note that $\Sq(0,0) = \Sq(1,1) = \Sq(0,1) = \Sq(1,0) = \{0,1\}$, and thus the set $\{P_{Y\widehat{Y}|X} \in \mathcal{M}_{\max}(\mathsf{q}): P_{Y|X} = \wyx\}$ can be characterized by 
\begin{align*}
\begin{bmatrix}
P_{Y\widehat{Y}|X}(00|0) = P_0(0)-r_1 & P_{Y\widehat{Y}|X}(01|0) = r_1\\
P_{Y\widehat{Y}|X}(10|0) = r'_1 & P_{Y\widehat{Y}|X}(11|0) = P_0(1) - r'_1
\end{bmatrix} \ \mathrm{and} \ 
\begin{bmatrix}
P_{Y\widehat{Y}|X}(00|1) = P_1(0) -r'_2 &  P_{Y\widehat{Y}|X}(01|1) = r'_2\\
P_{Y\widehat{Y}|X}(10|1) = r_2 & P_{Y\widehat{Y}|X}(11|1) = P_1(1) - r_2
\end{bmatrix}
\end{align*}
for $r_1 \in [0,P_0(0)], r'_1 \in [0,P_0(1)], r'_2 \in [0,P_1(0)], r_2 \in [0,P_1(1)]$. Setting $r_1 = r_2 = \frac{P_0(0)-P_1(0)}{2} = \frac{P_1(1) - P_0(1)}{2}$ and $r'_1 = r'_2 = 0$ yields $\DD(P_{\widehat{Y}|X=1} \Vert P_{\widehat{Y}|X=0}) = 0$. Therefore, the upper bound $\bar{\mathsf{R}}_{\mathsf{q},\delta} = 0$.

\section{Proof of Theorem~\ref{thm:eo}} \label{appendix:eo}

To evaluate the lower bound $\underline{\mathsf{R}}_{\mathsf{q},\delta}$, we substitute the decoding metric in~\eqref{eq:q34} to the expression in~\eqref{eq:lower}, yielding that
\begin{align}
    \underline{\mathsf{R}}_{\mathsf{q},\delta} &= t_{\delta} \cdot \left(\sup_{s \ge 0} \sum_{y \in \mathcal{Y}\setminus \mathcal{D}}P_1(y) \log \frac{\mathsf{q}(1,y)^s}{\mathsf{q}(0,y)^s} - \log \sum_{y \in \mathcal{Y}}P_0(y)\frac{\mathsf{q}(1,y)^s}{\mathsf{q}(0,y)^s} \right) \\
    &= t_{\delta} \cdot \Bigg(\sup_{s \ge 0} -\log \bigg[ P_0(\mathcal{Y}\setminus \mathcal{D}) +  \sum_{y \in \mathcal{D}}P_0(y) \cdot \xi^s \bigg] \Bigg) \label{eq:li} \\
    &= t_{\delta} \cdot \log\left(1/P_0(\mathcal{Y}\setminus \mathcal{D}) \right),  
\end{align}
where~\eqref{eq:li} holds since $\mathsf{q}(1,y) = \mathsf{q}(0,y) = 1$ for every $y \in \mathcal{Y}\setminus \mathcal{D}$, and $\mathsf{q}(1,y) = \xi$ and $\mathsf{q}(0,y) = 1$ for every $y \in \mathcal{D}$. This means that the rate $t_{\delta} \cdot \log\left(1/P_0(\mathcal{Y}\setminus \mathcal{D})\right)$ is achievable.

Next, we evaluate the upper bound $\bar{\mathsf{R}}_{\mathsf{q},\delta}$ in~\eqref{eq:minimizer}. Based on the decoding metric in~\eqref{eq:q34}, one can characterize the sets $\mathcal{S}_{\mathsf{q}}(y,\hat{y})$ for different pairs $(y,\hat{y}) \in \mathcal{Y} \times \mathcal{Y}$ as follows:
\begin{itemize}
    \item If $y \notin \mathcal{D}$ and $\hat{y} \in \mathcal{D}$, then $\mathcal{S}_{\mathsf{q}}(y,\hat{y}) = \{0\}$ and $\mathcal{S}_{\mathsf{q}}(\hat{y},y) = \{1\}$;
    \item If both $y \notin \mathcal{D}$ and $\hat{y} \notin \mathcal{D}$, then $\mathcal{S}_{\mathsf{q}}(y,\hat{y}) = \{0,1\}$;
    \item If both $y \in \mathcal{D}$ and $\hat{y} \in \mathcal{D}$, then $\mathcal{S}_{\mathsf{q}}(y,\hat{y}) = \{0,1\}$.
\end{itemize}
Given these sets $\mathcal{S}_{\mathsf{q}}(y,\hat{y})$, one can check that any joint conditional distribution $P_{Y\widehat{Y}|X}$ belonging to the set 
\begin{align}
\{P_{Y\widehat{Y}|X} \in \mathcal{M}_{\max}(\mathsf{q}): P_{Y|X} = \wyx\}   \label{eq:39}
\end{align}
also satisfies the following two properties:
\begin{align}
\sum_{\hat{y} \in \mathcal{Y}\setminus \mathcal{D}}P_{\widehat{Y}|X=1}(\hat{y}) = 1 \quad \mbox{and} \quad \sum_{\hat{y} \in \mathcal{Y}\setminus \mathcal{D}}P_{\widehat{Y}|X=0}(\hat{y}) \le  P_0(\mathcal{Y}\setminus\mathcal{D}). \label{eq:property}
\end{align}
Thus, for any joint conditional distribution $P_{Y\widehat{Y}|X}$ belonging to the set in~\eqref{eq:39}, we have 
\begin{align}
\DD\left(P_{\widehat{Y}|X=1} \Vert P_{\widehat{Y}|X=0}\right) &= \sum_{\hat{y}\in \mathcal{Y}} P_{\widehat{Y}|X=1}(\hat{y}) \log \frac{P_{\widehat{Y}|X=1}(\hat{y})}{P_{\widehat{Y}|X=0}(\hat{y})}\\
&= \sum_{\hat{y}\in \mathcal{Y} \setminus \mathcal{D}} P_{\widehat{Y}|X=1}(\hat{y}) \log \frac{P_{\widehat{Y}|X=1}(\hat{y})}{P_{\widehat{Y}|X=0}(\hat{y})} \label{eq:zou} \\
&= - \sum_{\hat{y}\in \mathcal{Y} \setminus \mathcal{D}} P_{\widehat{Y}|X=1}(\hat{y}) \log \frac{P_{\widehat{Y}|X=0}(\hat{y})}{P_{\widehat{Y}|X=1}(\hat{y})} \\
&\ge - \log \sum_{\hat{y}\in \mathcal{Y} \setminus \mathcal{D}} P_{\widehat{Y}|X=1}(\hat{y}) \frac{P_{\widehat{Y}|X=0}(\hat{y})}{P_{\widehat{Y}|X=1}(\hat{y})} \label{eq:jensen} \\
&\ge \log(1/P_0(\mathcal{Y}\setminus\mathcal{D})), \label{eq:last}
\end{align}
where~\eqref{eq:zou} follows from the first property stated in~\eqref{eq:property}, inequality~\eqref{eq:jensen} follows from Jensen's inequality, and inequality~\eqref{eq:last} is due to the second property stated  in~\eqref{eq:property}. Moreover, one can check that the following joint conditional distribution $P'_{Y\widehat{Y}|X}$, whose non-zero entries are given by
\begin{align}
    &\mbox{For a specific } y^{\ast} \notin \mathcal{D}: \ \left\{P'_{Y\widehat{Y}|X}(yy^{\ast}|0) = P_0(y)\right\}_{y \notin \mathcal{D}}, \ \left\{P'_{Y\widehat{Y}|X}(yy^{\ast}|1) = P_1(y)\right\}_{y \notin \mathcal{D}}, \\
    &\mbox{For a specific } y^{\dagger} \in \mathcal{D}: \ \left\{P'_{Y\widehat{Y}|X}(yy^{\dagger}|0) = P_0(y)\right\}_{y \in \mathcal{D}},
\end{align}
also belongs to the set in~\eqref{eq:39}, and this distribution satisfies  
\begin{align}
    \DD\left(P'_{\widehat{Y}|X=1} \Vert P'_{\widehat{Y}|X=0}\right) = \log(1/P_0(\mathcal{Y}\setminus\mathcal{D})). \label{eq:fa2}
\end{align}
Combining~\eqref{eq:last} and~\eqref{eq:fa2}, we note that the upper bound $\bar{\mathsf{R}}_{\mathsf{q},\delta}$ exactly equals $t_{\delta}\cdot \log(1/P_0(\mathcal{Y}\setminus\mathcal{D}))$, which matches the lower bound $\underline{\mathsf{R}}_{\mathsf{q},\delta}$.

\section{Proof of Lemma~\ref{lemma:expectation}} \label{appendix:expectation}
We first prove that the random variable $S_i$ is bounded. By symmetry, one can set $\underline{\X}^{(i)} = e_1$, where $e_j$ is a weight-one length-$w$ vector with the $j$-th element being $1$. We then rewrite $S_i$ as 
\begin{align}
S_i = -\log \sum_{j=1}^w \frac{1}{w} \frac{\mathsf{q}^w(e_j,\underline{\Y}^{(i)})^s}{\mathsf{q}^w(e_1,\underline{\Y}^{(i)})^s} &= -\log\left[ \frac{1}{w}\sum_{j=1}^w \frac{\mathsf{q}(0,\underline{Y}_1^{(i)})^s }{\mathsf{q}(1,\underline{Y}_1^{(i)})^s}\frac{\mathsf{q}(1,\underline{Y}_j^{(i)})^s }{\mathsf{q}(0,\underline{Y}_j^{(i)})^s} \right], \label{eq:57}
\end{align}
where $\underline{Y}_1^{(i)} \sim P_1$ and $\underline{Y}_j^{(i)} \sim P_0$ for $2 \le j \le w$ (since $\underline{\X}^{(i)} = e_1$). It is clear that for any realization of $\underline{\Y}^{(i)}$, \begin{align}
    -\log\left[ \left(\max_{y:P_1(y)>0} \frac{\mathsf{q}(0,y)^s}{\mathsf{q}(1,y)^s} \right) \left(\max_{y:P_0(y)>0} \frac{\mathsf{q}(1,y)^s}{\mathsf{q}(0,y)^s} \right) \right] \le S_i \le     -\log\left[ \left(\min_{y:P_1(y)>0} \frac{\mathsf{q}(0,y)^s}{\mathsf{q}(1,y)^s} \right) \left(\min_{y:P_0(y)>0} \frac{\mathsf{q}(1,y)^s}{\mathsf{q}(0,y)^s} \right) \right], 
\end{align}
which means that $S_i$ is bounded since $\mathsf{q}$ takes on values on the positive real line.

Next, we calculate the expectation of $S_i$:
\begin{align}
    \mathbb{E}(S_i) &= \frac{1}{w}\sum_{k=1}^w \sum_{\underline{y}^{(i)}} \wyxw(\underline{y}^{(i)}|e_k) \cdot \left[-\log \sum_{j=1}^w \frac{1}{w}\frac{\mathsf{q}^w(e_j,\underline{y}^{(i)})^s}{\mathsf{q}^w(e_k, \underline{y}^{(i)})^s} \right] \\
    &= \sum_{\underline{y}^{(i)}} \wyxw(\underline{y}^{(i)}|e_1) \cdot \left[-\log \sum_{j=1}^w \frac{1}{w}\frac{\mathsf{q}^w(e_j,\underline{y}^{(i)})^s}{\mathsf{q}^w(e_1, \underline{y}^{(i)})^s} \right] \label{eq:sym}\\
    &= - \sum_{\underline{y}^{(i)}_1}P_1(\underline{y}^{(i)}_1) \sum_{(\underline{y}^{(i)})_2^w} P_0^{\otimes w-1}\left((\underline{y}^{(i)})_2^w\right) \cdot  \log \frac{1}{w}\sum_{j=1}^w \frac{\mathsf{q}(0,\underline{y}^{(i)}_1)^s}{\mathsf{q}(1,\underline{y}^{(i)}_1)^s} \frac{\mathsf{q}(1,\underline{y}^{(i)}_j)^s}{\mathsf{q}(0,\underline{y}^{(i)}_j)^s} \\
    &= \left(-\sum_{\underline{y}^{(i)}_1}P_1(\underline{y}^{(i)}_1) \log\frac{\mathsf{q}(0,\underline{y}^{(i)}_1)^s}{\mathsf{q}(1,\underline{y}^{(i)}_1)^s} \right) -  \left(\sum_{\underline{y}^{(i)}_1}P_1(\underline{y}^{(i)}_1) \sum_{(\underline{y}^{(i)})_2^w} P_0^{\otimes w-1}\left((\underline{y}^{(i)})_2^w\right) \log\frac{1}{w}\sum_{j=1}^w \frac{\mathsf{q}(1,\underline{y}^{(i)}_j)^s}{\mathsf{q}(0,\underline{y}^{(i)}_j)^s}  \right), \label{eq:jian}
\end{align}
where~\eqref{eq:sym} is due to symmetry. Note that the first term in~\eqref{eq:jian} equals $\sum_{y}P_1(y) \log\frac{\mathsf{q}(1,y)^s}{\mathsf{q}(0,y)^s}$. The second term in~\eqref{eq:jian} can be bounded using Jensen's inequality as follows:
\begin{align}
&\sum_{\underline{y}^{(i)}_1}P_1(\underline{y}^{(i)}_1)\sum_{(\underline{y}^{(i)})_2^w} P_0^{\otimes w-1}\left((\underline{y}^{(i)})_2^w\right) \log\frac{1}{w}\sum_{j=1}^w \frac{\mathsf{q}(1,\underline{y}^{(i)}_j)^s}{\mathsf{q}(0,\underline{y}^{(i)}_j)^s} \notag \\
&\le \log \left[ \sum_{\underline{y}^{(i)}_1}P_1(\underline{y}^{(i)}_1) \sum_{(\underline{y}^{(i)})_2^w} P_0^{\otimes w-1}\left((\underline{y}^{(i)})_2^w\right) \frac{1}{w}\sum_{j=1}^w \frac{\mathsf{q}(1,\underline{y}^{(i)}_j)^s}{\mathsf{q}(0,\underline{y}^{(i)}_j)^s}  \right] \\
&= \log \left[\frac{w-1}{w} \left( \sum_{y}P_0(y) \frac{\mathsf{q}(1,y)^s}{\mathsf{q}(0,y)^s} \right) + \frac{1}{w}\left( \sum_{y}P_1(y) \frac{\mathsf{q}(1,y)^s}{\mathsf{q}(0,y)^s} \right) \right] \\
&\le \log \left( \sum_{y}P_0(y) \frac{\mathsf{q}(1,y)^s}{\mathsf{q}(0,y)^s} \right) + \frac{C_0}{w} \label{eq:qi}
\end{align}
for some constant $C_0 > 0$.
Combining~\eqref{eq:jian} and~\eqref{eq:qi}, we complete the proof of Lemma~\ref{lemma:expectation}.

\section{Proof of Lemma~\ref{lemma:subset}} \label{appendix:b}

    First, the probability of $\mathcal{E}_1$ can be bounded from below as
    \begin{align}
        \PP(\mathcal{E}_1) &= \PP\left( \bigcap_{k \in \mathcal{K}} \left\{P^{\mathrm{(avg)}}_{\mathrm{err}}(\C_k) \le 2n\exp(-n^{1/4}) \right\} \right)  \\
        &= \left[1- \PP\left(P^{\mathrm{(avg)}}_{\mathrm{err}}(\C_1) > 2n\exp(-n^{1/4}) \right) \right]^{|\mathcal{K}|} \label{eq:feng1} \\
        &\ge \left(1 - \frac{1}{n} \right)^{|\mathcal{K}|}, \label{eq:feng2}
    \end{align}
where~\eqref{eq:feng1} holds since the sub-codes are generated independently, and~\eqref{eq:feng2} follows from Markov's inequality and the fact that $\mathbb{E}(P^{\mathrm{(avg)}}_{\mathrm{err}}(\C_1)) \le 2\exp(-n^{1/4})$. The rest of the proof essentially follows from that of~\cite[Lemma 4]{tahmasbi2018first}. As shown in~\cite[Eqns.(94)-(100)]{tahmasbi2018first}, the probability of $\mathcal{E}_2$ can be bounded from below as  
\begin{align}
   \PP(\mathcal{E}_2) \ge  1 - \exp \left\{-|\mathcal{M}||\mathcal{K}| \left( \frac{2 \lambda_1 \lambda_2^2}{\log^2\left(\frac{\lambda_1 |\mathcal{M}||\mathcal{K}||\mathcal{Z}^n|}{([\ppmz]^{\min})^2} \right)} - H_b(\lambda_1) \right) \right\}.
\end{align}
Therefore, 
\begin{align}
\PP(\mathcal{E}_1 \cap \mathcal{E}_2) \ge 1 - \PP(\mathcal{E}_1^c) - \PP(\mathcal{E}_2^c) &= \PP(\mathcal{E}_1) + \PP(\mathcal{E}_2) - 1 \\ 
&\ge \left(1 - \frac{1}{n} \right)^{|\mathcal{K}|} - \exp \left\{-|\mathcal{M}||\mathcal{K}| \left(\frac{2 \lambda_1 \lambda_2^2}{\log^2\left(\frac{\lambda_1 |\mathcal{M}||\mathcal{K}||\mathcal{Z}^n|}{([\ppmz]^{\min})^2} \right)} - H_b(\lambda_1) \right) \right\}, 
\end{align}
which is positive when the condition in~\eqref{eq:huo1} holds.

\section{Proof of Lemma~\ref{lemma:typical}} \label{appendix:typical}
Without loss of generality, we consider a specific codeword $\x \in \C_i^{cc}$ such that its type $T_{\x} = P^{cc}$, $x_1 = \cdots = x_{t'\sqrt{n}} = 1$, and $x_{t'\sqrt{n}+1} = \cdots = x_n = 0$.  For every $y \in \mathcal{Y}$, the expected value of $\sum_{j=1}^{t'\sqrt{n}}\mathbbm{1}\{(x_j,Y_j)=(1,y)\}$ is $t'\sqrt{n}P_1(y)$, and by applying the Chernoff bound we have that for any $\epsilon_1 > 0$,
\begin{align}
    \PP\left( \left|\sum_{j=1}^{t'\sqrt{n}}\mathbbm{1}\{(x_j,Y_j)=(1,y)\} - t'\sqrt{n}P_1(y) \right| \ge \epsilon_1 t'\sqrt{n}P_1(y) \right) \le 2 \exp \left\{-\frac{1}{3} \epsilon_1^2 t'\sqrt{n}P_1(y) \right\}.
\end{align}
Setting $\epsilon_1 = n^{-1/8}$, we obtain that with probability at least $1 - 2\exp\left\{-\frac{1}{3}t'n^{1/4}[P_1]^{\min} \right\}$, 
\begin{align}
\left|T_{\Y|\x}(y|1) - P_1(y) \right| =   \left|\frac{\left|\sum_{j=1}^{t'\sqrt{n}}\mathbbm{1}\{(x_j,Y_j)=(1,y)\}\right|}{t'\sqrt{n}} - P_1(y) \right| \le P_1(y) n^{-1/8} \le  [P_1]^{\max}n^{-1/8}. \label{eq:app1}
\end{align}

Similarly, for every $y \in \mathcal{Y}$, the expected  value of $\sum_{j=t'\sqrt{n}+1}^{n} \mathbbm{1}\{(x_j,Y_j)=(0,y)\}$ is $(n-t'\sqrt{n})P_0(y)$, and by applying the Chernoff bound we have that for any $\epsilon_2 > 0$, 
\begin{align}
    \PP\left( \left|\sum_{j=t'\sqrt{n}+1}^{n}\mathbbm{1}\{(x_j,Y_j)=(0,y)\} - (n-t'\sqrt{n})P_0(y) \right| \ge \epsilon_2 (n-t'\sqrt{n})P_0(y) \right) \le 2 \exp \left\{-\frac{1}{3} \epsilon_2^2 (n-t'\sqrt{n})P_0(y) \right\}. 
\end{align}
Setting $\epsilon_2 = \sqrt{(\log n)/n}$, we obtain that with probability at least $1 - 2n^{-\frac{1}{3} [P_0]^{\min}}$, 
\begin{align}
\left|T_{\Y|\x}(y|0) - P_0(y) \right| =   \left|\frac{\left|\sum_{j=t'\sqrt{n}+1}^{n}\mathbbm{1}\{(x_j,Y_j)=(0,y)\}\right|}{n - t'\sqrt{n}} - P_0(y) \right| \le P_0(y) \sqrt{\frac{\log n}{n}} \le [P_0]^{\max} \sqrt{\frac{\log n}{n}}. \label{eq:app2}
\end{align}
Combining~\eqref{eq:app1} and~\eqref{eq:app2}, we complete the proof of Lemma~\ref{lemma:typical}.

\section{Proof of Lemma~\ref{lemma:a}} \label{appendix:a}
First, we construct a joint conditional type $\bar{V}_{Y\widehat{Y}|X}$ that corresponds to the joint conditional distribution $\Pj^{\ast}$. For every $x \in \mathcal{X}$ and $y,\yh \in \mathcal{Y}$, we set $\bar{V}_{Y\widehat{Y}|X}(y,\yh|x)$ to be either 
\begin{align}
\frac{\left\lfloor nP^{cc}(x)\Pj^\ast(y,\yh|x)\right\rfloor}{nP^{cc}(x)} \quad \mathrm{or} \quad \frac{\left\lceil nP^{cc}(x)\Pj^\ast(y,\yh|x)\right\rceil}{nP^{cc}(x)}   
\end{align}
such that $\sum_{y,\yh}\bar{V}_{Y\widehat{Y}|X}(y,\yh|x) = 1$. By construction, $\bar{V}_{Y\widehat{Y}|X}$ is a maximal joint conditional type, and every pair $(y,\yh) \in \mathcal{Y} \times \mathcal{Y}$ satisfies 
\begin{align}
\left|\bar{V}_{Y\widehat{Y}|X}(y,\yh|x) -  \Pj^\ast(y,\yh|x)\right| \le  \frac{1}{nP^{cc}(x)}.   
\end{align}
Furthermore, if we consider the marginals of 
$\bar{V}_{Y\widehat{Y}|X}$ and $\Pj^\ast$, we have
\begin{align}
    &\left| \bar{V}_{Y|X}(y|x) - P_{Y|X}^\ast(y|x) \right| \le \sum_{\yh} \left| \bar{V}_{Y\widehat{Y}|X}(y,\yh|x) - \Pj^\ast(y,\yh|x) \right| \le \frac{|\mathcal{Y}|}{nP^{cc}(x)}, \\
    &\left| \bar{V}_{\widehat{Y}|X}(\yh|x) - P_{\widehat{Y}|X}^\ast(\yh|x) \right| \le \sum_{y} \left| \bar{V}_{Y\widehat{Y}|X}(y,\yh|x) - \Pj^\ast(y,\yh|x) \right| \le \frac{|\mathcal{Y}|}{nP^{cc}(x)}. \label{eq:15}
\end{align}
For simplicity we define $\kappa_n(0) \triangleq \frac{|\mathcal{Y}|}{n - t'\sqrt{n}} + [P_0]^{\max}\sqrt{\frac{\log n}{n}}$ and $\kappa_n(1) \triangleq \frac{|\mathcal{Y}|}{t'\sqrt{n}} + [P_1]^{\max}n^{-1/8}$. Since the conditional type $V' \in \mathscr{V}_{P^{cc}}$ and by noting that $P_{Y|X}^\ast = \wyx$, we have that for all $y \in \mathcal{Y}$,
\begin{align}
   \left| \bar{V}_{Y|X}(y|x) - V'(y|x) \right| \le \left| \bar{V}_{Y|X}(y|x) - \wyx(y|x) \right| + \left| \wyx(y|x) - V'(y|x) \right|  \le \kappa_n(x).
\end{align}
We now construct another joint conditional type $V_{Y\widehat{Y}|X}$ as follows: 
\begin{itemize}
\item For $y \in \mathcal{Y}$ such that $\bar{V}_{Y|X}(y|x) \le V'(y|x)$, we add non-negative real numbers $\{\epsilon(x,y,\yh)\}_{\yh \in \mathcal{Y}}$ to the elements in $\{\bar{V}_{Y\widehat{Y}|X}(y,\yh|x)\}_{\yh \in \mathcal{Y}}$ such that $V_{Y\widehat{Y}|X}(y,\yh|x) \triangleq \bar{V}_{Y\widehat{Y}|X}(y,\yh|x) + \epsilon(x,y,\yh)$. We choose $\{\epsilon(x,y,\yh)\}_{\yh \in \mathcal{Y}}$ in such a way that (i) $0 \le \epsilon(x,y,\yh) \le \kappa_n(x)$, (ii) $\epsilon(x,y,\yh) = 0$ if $\bar{V}_{Y\widehat{Y}|X}(y,\yh|x) = 0$, and (iii) $\sum_{\yh} \epsilon(x,y,\yh) = V'(y|x) - \bar{V}_{Y|X}(y|x)$. 
    
\item For $y \in \mathcal{Y}$ such that $\bar{V}_{Y|X}(y|x) > V'(y|x)$, we add non-positive real numbers $\{\epsilon(x,y,\yh)\}_{\yh \in \mathcal{Y}}$ to the elements in $\{\bar{V}_{Y\widehat{Y}|X}(y,\yh|x)\}_{\yh \in \mathcal{Y}}$ such that $V_{Y\widehat{Y}|X}(y,\yh|x) \triangleq \bar{V}_{Y\widehat{Y}|X}(y,\yh|x) + \epsilon(x,y,\yh)$. We choose $\{\epsilon(x,y,\yh)\}_{\yh \in \mathcal{Y}}$ in such a way that (i) $-\kappa_n(x) \le \epsilon(x,y,\yh) \le 0$, (ii) $\bar{V}_{Y\widehat{Y}|X}(y,\yh|x) + \epsilon(x,y,\yh) \ge 0$ for all $\yh \in \mathcal{Y}$, and (iii) $\sum_{\yh} \epsilon(x,y,\yh) = V'(y|x) - \bar{V}_{Y|X}(y|x)$. 
\end{itemize}
The properties of $\{\epsilon(x,y,\yh)\}_{\yh \in \mathcal{Y}}$ ensure that (i) $V_{Y\widehat{Y}|X}$ is a maximal joint conditional type and (ii) the marginal conditional type $V_{Y|X}$ satisfies
\begin{align}
V_{Y|X}(y|x) = \sum_{\yh} V_{Y\widehat{Y}|X}(y,\yh|x) = V'(y|x), \quad \forall x \in \mathcal{X}, y \in \mathcal{Y}.
\end{align}
Finally, we examine the other marginal conditional type $V_{\widehat{Y}|X}$ as follows:
\begin{align}
\left|V_{\widehat{Y}|X}(\yh|x) - P^{\ast}_{\widehat{Y}|X}(\yh|x) \right| &\le \sum_{y} \left|V_{Y\widehat{Y}|X}(y,\yh|x) - P^{\ast}_{Y\widehat{Y}|X}(y,\yh|x) \right|    \\
&\le \sum_{y} \left|\bar{V}_{Y\widehat{Y}|X}(y,\yh|x) - P^{\ast}_{Y\widehat{Y}|X}(y,\yh|x)\right| + \sum_{y} |\epsilon(x,y,\yh)| \\
&\le \frac{|\mathcal{Y}|}{nP^{cc}(x)} + |\mathcal{Y}|\cdot\kappa_n(x), \label{eq:20}
\end{align}
where~\eqref{eq:20} follows from the inequality in~\eqref{eq:15} as well as the fact that $|\epsilon(x,y,\yh)| \le \kappa_n(x)$. This completes the proof of Lemma~\ref{lemma:a}.

\section{Proof of Lemma~\ref{lemma:key}} \label{appendix:key}

We label all the elements in $\widehat{\mathscr{V}}_{P^{cc}}$ by $\left\{\Vh^{1}, \Vh^{2}, \ldots, \Vh^{|\widehat{\mathscr{V}}_{P^{cc}}|}\right\}$. The first step is to show that there exists a codeword $\x(m) \in \C_i^{cc}$ and a collection of sets $\B_j \subset \T_{\Vh^j}(\x(m))$ for $j \in \{1,2,\ldots, |\widehat{\mathscr{V}}_{P^{cc}}|\}$ such that
\begin{enumerate}
    \item $|\B_j| \ge \frac{a-1}{a}\left|\T_{\Vh^j}(\x(m)) \right|$;
    \item For all $\yyh \in \B_j$, the code $\C_i^{cc}$ contains $a$ other codewords $\x'(1), \ldots, \x'(a)$ that have exactly the same conditional type with $\x(m)$, i.e., $T_{\yyh|\x'(1)} = \cdots = T_{\yyh|\x'(a)} = T_{\yyh|\x(m)}$. 
\end{enumerate}
We prove the above argument by contradiction. Suppose for every $\x \in \C_i^{cc}$, there is a set $\A_{\x} \subset \T_{\Vh^j}(\x)$ for some $j \in \{1,2,\ldots, |\widehat{\mathscr{V}}_{P^{cc}}|\}$ such that 
\begin{enumerate}
    \item $|\A_{\x}| > \frac{1}{a}\left|\T_{\Vh^j}(\x) \right|$;
    \item For all $\yyh \in \A_{\x}$, there are at most $a-1$ other codewords $\x'(1), \ldots, \x'(a-1) \in \C_i^{cc}$ satisfying $T_{\yyh|\x'(1)} = \cdots = T_{\yyh|\x'(a-1)} = T_{\yyh|\x}$.
\end{enumerate}
We then partition the constant composition code $\C_i^{cc}$ into $|\widehat{\mathscr{V}}_{P^{cc}}|$ disjoint subsets $\C_i^{cc}(j)$, where $\C_i^{cc}(j) \triangleq \left\{\x \in \C_i^{cc}: \A_{\x} \subset \T_{\Vh^j}(\x) \right\}.$
For every $\yyh \in \T_{P^{cc}\Vh^j}$, one can prove that it is a member of at most $a$ sets $\{\A_{\x}\}_{\x \in \C_i^{cc}(j)}$ that corresponds to codewords in a single $\C_i^{cc}(j)$. This is because if $\yyh$ were belonging to $\A_{\x(1)}\cap\A_{\x(2)}\cap\cdots\cap \A_{\x(a+1)}$ (where $\x(1),\x(2),\ldots, \x(a+1) \in \C_i^{cc}(j)$), then such $\yyh$ would satisfy $\yyh\in \A_{\x(1)}$ but there are $a$ other codewords $\x(2), \ldots, \x(a+1) \in \C_i^{cc}$ such that $T_{\y|\x(1)} = T_{\yyh|\x(2)} = \ldots = T_{\yyh|\x(a+1)}$, thus violating the second property that $\A_{\x(1)}$ should satisfy. Then for every $j \in \{1,2,\ldots, |\widehat{\mathscr{V}}_{P^{cc}}|\}$, we have 
\begin{align}
    \sum_{\x \in \C_i^{cc}(j)} |\A_{\x}| &=  \sum_{\x \in \C_i^{cc}(j)} \sum_{\yyh \in \T_{P^{cc}\Vh^j}} \mathbbm{1}\left\{\yyh \in \A_{\x} \right\} =\sum_{\yyh \in \T_{P^{cc}\Vh^j}} \sum_{\x \in \C_i^{cc}(j)} \mathbbm{1}\left\{\yyh \in \A_{\x} \right\} \le  a \times \left|\T_{P^{cc}\Vh^j} \right|. \label{eq:contra1}
\end{align}
On the other hand, by the first property of $\A_{\x}$ we have
\begin{align}
\sum_{\x \in \C_i^{cc}(j)} |\A_{\x}| > \sum_{\x \in \C_i^{cc}(j)} \frac{1}{a} \left|\T_{\Vh^j}(\x) \right| &\ge \sum_{\x \in \C_i^{cc}(j)} \frac{|\widehat{\mathscr{V}}_{P^{cc}}|}{|\C_i^{cc}|}\times a \times \left|\T_{P^{cc}\Vh^j} \right| \\
&=  \frac{|\widehat{\mathscr{V}}_{P^{cc}}||\C_i^{cc}(j)|}{|\C_i^{cc}|}\times a \times \left|\T_{P^{cc}\Vh^j} \right|. 
\end{align}
By the Pigeonhole principle, there must exist a $\C_i^{cc}(j')$ (where $j' \in \{1,2,\ldots, |\widehat{\mathscr{V}}_{P^{cc}}|\}$) such that $|\C_{i}^{cc}(j')| \ge |\C_i^{cc}|/|\widehat{\mathscr{V}}_{P^{cc}}|$, thus 
\begin{align}
   \sum_{\x \in \C_i^{cc}(j')} |\A_{\x}| > a \times \left|\T_{P^{cc}\Vh^{j'}} \right|.  \label{eq:contra2}
\end{align}
Therefore, a contradiction for $\C_i^{cc}(j')$ arises due to~\eqref{eq:contra1} and~\eqref{eq:contra2}. This means that when~\eqref{eq:condition} holds, there must exist a ``bad'' codeword $\x(m)$ such that for all conditional type $\Vh^j \in \widehat{\mathscr{V}}_{P^{cc}}$, a large fraction of $\yyh$ sequences in the corresponding type class cause a \emph{type conflict} with $a$ other codewords (i.e., there are $a$ codewords $\x'(1), \ldots, \x'(a)$ satisfying $T_{\yyh|\x'(1)} = \cdots = T_{\yyh|\x'(a)} = T_{\yyh|\x(m)}$). 

The rest of the proof is essentially due to~\cite{kangarshahi2021single}. Now, suppose the aforementioned codeword $\x(m)$ is transmitted, and Bob's received sequence $\Y$ belongs to a conditional type $V \in \mathscr{V}_{P^{cc}}$ (which occurs with high probability). By Lemma~\ref{lemma:a}, one can find a maximal joint conditional type $V_{Y\widehat{Y}|X}$ such that its marginal conditional types $V_{Y|X} = V$ and $V_{\widehat{Y}|X} \in \widehat{\mathscr{V}}_{P^{cc}}$. As noted in~\cite{kangarshahi2021single}, any pair of sequences $(\y,\yyh)$ in the maximal joint conditional type has the following property.
\begin{claim}[Lemma 3 of~\cite{kangarshahi2021single}] \label{lemma:link}
If $(\y,\yyh) \in \mathcal{T}_{V_{Y\widehat{Y}|X}}(\x(m))$, and $\yyh$ has a type conflict with another codeword $\x' \in \C_i^{cc}$, i.e., $T_{\yyh|\x'} = T_{\yyh|\x(m)} = V_{\widehat{Y}|X}$, then $\mathsf{q}^n(\x,\y) \le \mathsf{q}^n(\x',\y)$.
\end{claim}
Claim~\ref{lemma:link} implies that for any $\y \in \mathcal{T}_{V}(\x(m))$, if one can find a $\yyh$ such that (i) $(\y,\yyh) \in \mathcal{T}_{V_{Y\widehat{Y}|X}}(\x(m))$, and (ii) $\yyh$ has a type conflict with $a$ other codewords, then the probability of decoding error when receiving $\y$ (under the decoding metric $\mathsf{q}$) is at least $\frac{a}{a+1}$. In fact,~\cite[Lemma 2]{kangarshahi2021single} shows that at least $\frac{a-1}{a}$ fraction of sequences $\y \in\mathcal{T}_V(\x(m))$ satisfies the above two properties simultaneously.\footnote{Roughly speaking, this is because for the ``bad'' codeword $\x(m)$, most sequences (of fraction $\frac{a-1}{a}$) in $\mathcal{T}_{V_{\widehat{Y}|X}}$ have a type conflict with at least $a$ other codewords.} Thus, when Bob's received sequence $\Y \in V$ for some $V \in \mathscr{V}_{P^{cc}}$, we have 
\begin{align}
        \PP\left(\widehat{M} \ne m|M = m, T_{\Y|\x(m)} =V \right) \ge \frac{a-1}{a+1} = 1 - \frac{2}{a+1}.
    \end{align}

\section{Proof of Eqn.~\eqref{eq:c}} \label{appendix:calculation1}
By using Taylor expansions, we have
\begin{align}
&H\left(\frac{t'}{\sqrt{n}}\Vh_1 + \left(1 - \frac{t'}{\sqrt{n}} \right) \Vh_0 \right) - \frac{t'}{\sqrt{n}}H(\Vh_1) - \left(1- \frac{t'}{\sqrt{n}} \right)H(\Vh_0) \notag \\
&= -\left(1 - \frac{t'}{\sqrt{n}} \right) \sum_{y}\Vh_0(y) \log \left[1 - \frac{t'}{\sqrt{n}} + \frac{t}{\sqrt{n}}\frac{\Vh_1(y)}{\Vh_0(y)} \right] - \frac{t'}{\sqrt{n}} \sum_{y}\Vh_1(y)\log\left[\frac{t'}{\sqrt{n}} + \left(1 -\frac{t'}{\sqrt{n}} \right)\frac{\Vh_0(y)}{\Vh_1(y)} \right]  \\
&= -\left(1 - \frac{t'}{\sqrt{n}} \right) \sum_{y}\Vh_0(y) \left[\frac{t'}{\sqrt{n}}\frac{\Vh_1(y) - \Vh_0(y)}{\Vh_0(y)} + \mathcal{O}\left(\frac{(t')^2}{n}\right) \right] \notag \\
&\qquad\qquad\qquad\qquad\qquad\qquad\qquad\qquad - \frac{t'}{\sqrt{n}} \sum_{y}\Vh_1(y)\log\left[\frac{\Vh_0(y)}{\Vh_1(y)}\left(1 + \frac{t'}{\sqrt{n}}\left(1 - \frac{\Vh_1(y)}{\Vh_0(y)}\right) \right) \right] \\
&\le \frac{t'}{\sqrt{n}} \DD\left(\Vh_1 \big\Vert \Vh_0 \right) + \frac{C_4 \cdot (t')^2}{n}
\end{align}
for some constant $C_4 > 0$.

\section{Proof of Eqn.~\eqref{eq:xi2}} \label{appendix:d}
For simplicity we define $\tau^1_n \triangleq \frac{[P_1]^{\max}|\mathcal{Y}|}{n^{1/8}}+ \frac{|\mathcal{Y}|^2+|\mathcal{Y}|}{t'\sqrt{n}}$ and $\tau^0_n \triangleq \frac{[P_0]^{\max}|\mathcal{Y}| \sqrt{\log n}}{\sqrt{n}} + \frac{2|\mathcal{Y}|^2+|2\mathcal{Y}|}{n}$.
For each $\Vh \in \widehat{\mathscr{V}}_{P^{cc}}$, we have
\begin{align}
\DD\left(\Vh_1 \big\Vert \Vh_0 \right) &= \sum_{\yh:\Vh_1(\yh) \ge \Vh_0(\yh)} \Vh_1(\yh) \log \frac{\Vh_1(\yh)}{\Vh_0(\yh)} + \sum_{\yh:\Vh_1(\yh) < \Vh_0(\yh)} \Vh_1(\yh) \log \frac{\Vh_1(\yh)}{\Vh_0(\yh)} \\
&\le \sum_{\yh:\Vh_1(\yh) \ge \Vh_0(\yh)} \left( P^{\ast}_{\widehat{Y}|X}(\yh|1) + \tau^1_n \right) \log \frac{P^{\ast}_{\widehat{Y}|X}(\yh|1)+\tau^1_n}{P^{\ast}_{\widehat{Y}|X}(\yh|0)-\tau^0_n} + \sum_{\yh:\Vh_1(\yh) < \Vh_0(\yh)} \left( P^{\ast}_{\widehat{Y}|X}(\yh|1) - \tau^1_n \right) \log \frac{P^{\ast}_{\widehat{Y}|X}(\yh|1)+ \tau^1_n}{P^{\ast}_{\widehat{Y}|X}(\yh|0)-\tau^0_n} \label{eq:bye3} \\
&\le \sum_{\yh\in \mathcal{Y}}P^{\ast}_{\widehat{Y}|X}(\yh|1) \log\frac{P^{\ast}_{\widehat{Y}|X}(\yh|1)+\tau^1_n}{P^{\ast}_{\widehat{Y}|X}(\yh|0)-\tau^0_n} + \tau^1_n \sum_{\yh \in \mathcal{Y}} \left|\log \frac{P^{\ast}_{\widehat{Y}|X}(\yh|1)+\tau^1_n}{P^{\ast}_{\widehat{Y}|X}(\yh|0)-\tau^0_n} \right|, \label{eq:bye1}
\end{align}
where~\eqref{eq:bye3} follows from the definition of $\widehat{\mathscr{V}}_{P^{cc}}$.
By applying Taylor expansions, we have
\begin{align}
\log\frac{P^{\ast}_{\widehat{Y}|X}(\yh|1)+\tau^1_n}{P^{\ast}_{\widehat{Y}|X}(\yh|0)-\tau^0_n} \le \log \frac{P^{\ast}_{\widehat{Y}|X}(\yh|1)}{P^{\ast}_{\widehat{Y}|X}(\yh|0)} + C_5 \tau^1_n + C_6 \tau^0_n \label{eq:bye2}
\end{align}
for some constants $C_5, C_6 > 0$. Combining~\eqref{eq:bye1} and~\eqref{eq:bye2}, we eventually obtain that there exists a vanishing sequence $\mu_n$ (depending on $\tau^1_n$ and $\tau^0_n$) such that \begin{align}
    \DD\left(\Vh_1 \big\Vert \Vh_0 \right) \le \DD\left(P^{\ast}_{\widehat{Y}|X=1} \Vert P^{\ast}_{\widehat{Y}|X=0} \right) + \mu_n.
\end{align}

\ifCLASSOPTIONcaptionsoff
  \newpage
\fi

\bibliographystyle{IEEEtran}
\bibliography{reference}

\end{document}